\documentclass[a4paper,12pt]{amsart} 
\usepackage{amsmath,amssymb,amsthm}
\usepackage{amscd}
\usepackage{mathrsfs}
\usepackage[usenames,dvipsnames]{color}
\usepackage{pstricks}
\usepackage{graphicx}
\usepackage{extarrows}
\usepackage{tabularx}
\usepackage{slashed}
\usepackage{tikz}
\usepackage{xparse}
\usepackage{enumitem}
\usepackage[normalem]{ulem}

\usepackage{ifpdf}
\ifpdf
  \usepackage[colorlinks,final,backref=page,hyperindex]{hyperref}
\else
  \usepackage[colorlinks,final,backref=page,hyperindex,hypertex]{hyperref}
\fi

\usepackage{xspace}

\usetikzlibrary{arrows,matrix}
\usetikzlibrary{decorations.pathmorphing}

\def\XXint#1#2#3{{\setbox0=\hbox{$#1{#2#3}{\int}$ }
\vcenter{\hbox{$#2#3$ }}\kern-.58\wd0}}

\def\XXsum#1#2#3{{\setbox0=\hbox{$#1{#2#3}{\sum}$ }
\vcenter{\hbox{$#2#3$ }}\kern-.51\wd0}}

\begin{document}

\newcommand\cutoffint{\mathop{-\hskip -4mm\int}\limits}
\newcommand\cutoffsum{\mathop{-\hskip -4mm\sum}\limits}
\newcommand\cutoffzeta{-\hskip -1.7mm\zeta} 
\newcommand{\goth}[1]{\ensuremath{\mathfrak{#1}}}
\newcommand{\bbox}{\normalsize {}%
        \nolinebreak \hfill $\blacksquare$ \medbreak \par}
\newcommand{\simall}[2]{\underset{#1\rightarrow#2}{\sim}}

\newtheorem{theorem}{Theorem}[section]
\newtheorem{lem}[theorem]{Lemma}
\newtheorem{coro}[theorem]{Corollary}
\newtheorem{problem}[theorem]{Problem}
\newtheorem{conjecture}[theorem]{Conjecture}
\newtheorem{prop}[theorem]{Proposition}
\newtheorem{propdefn}[theorem]{Proposition-Definition}
\newtheorem{lemdefn}[theorem]{Lemma-Definition}
\theoremstyle{definition}
\newtheorem{defn}[theorem]{Definition}
\newtheorem{rk}[theorem]{Remark}
\newtheorem{ex}[theorem]{Example}
\newtheorem{coex}[theorem]{Counterexample}
\newtheorem{algorithm}[theorem]{Algorithm}
\newtheorem{convention}[theorem]{Convention}
\newtheorem{principle}[theorem]{Principle}

\renewcommand{\theenumi}{{\it\roman{enumi}}}
\renewcommand{\theenumii}{{\alph{enumii}}}

\newenvironment{thmenumerate}
{\leavevmode\begin{enumerate}[leftmargin=1.5em]}{\end{enumerate}}

\newcommand{\nc}{\newcommand}
\newcommand{\delete}[1]{}
\newcommand{\aside}[1]{\delete{#1}}

\nc{\mlabel}[1]{\label{#1}}  
\nc{\mcite}[1]{\cite{#1}}  
\nc{\mref}[1]{\ref{#1}}  
\nc{\mbibitem}[1]{\bibitem{#1}} 

\delete{
\nc{\mlabel}[1]{\label{#1}  
{\hfill \hspace{1cm}{\bf{{\ }\hfill(#1)}}}}
\nc{\mcite}[1]{\cite{#1}{{\bf{{\ }(#1)}}}}  
\nc{\mref}[1]{\ref{#1}{{\bf{{\ }(#1)}}}}  
\nc{\mbibitem}[1]{\bibitem[\bf #1]{#1}} 
}

\newcommand{\bottop}{\top\hspace{-0.8em}\bot}
\nc{\mtop}{\top\hspace{-1mm}}
\nc{\tforall}{\text{ for all }}
\nc{\bfcf}{{\calc}}
\newcommand{\R}{\mathbb{R}}
\newcommand{\FC}{\mathbb{C}}
\newcommand{\K}{\mathbb{K}}
\newcommand{\Z}{\mathbb{Z}}
\nc{\PP}{\mathbb{P}}
\newcommand{\Q}{\mathbb{Q}}
\newcommand{\N}{\mathbb{N}}
\newcommand{\F}{\mathbb{F}}
\newcommand{\T}{\mathbb{T}}
\newcommand{\G}{\mathbb{G}}
\newcommand{\C}{\mathbb{C}}

\newcommand{\bfc}{\calc}

\newcommand {\frakb}{{\mathfrak {b}}}
\newcommand {\frakc}{{\mathfrak {c}}}
\newcommand {\frakd}{{\mathfrak {d}}}
\newcommand {\fraku}{{\mathfrak {u}}}
\newcommand {\fraks}{{\mathfrak {s}}}

\newcommand {\cala}{{\mathcal {A}}}
\newcommand {\calb}{\mathcal {B}}
\newcommand {\calc}{{\mathcal {C}}}
\newcommand {\cald}{{\mathcal {D}}}
\newcommand {\cale}{{\mathcal {E}}}
\newcommand {\calf}{{\mathcal {F}}}
\newcommand {\calg}{{\mathcal {G}}}
\newcommand {\calh}{\mathcal{H}}
\newcommand {\cali}{\mathcal{I}}
\newcommand {\call}{{\mathcal {L}}}
\newcommand {\calm}{{\mathcal {M}}}
\newcommand {\calp}{{\mathcal {P}}}
\newcommand {\calv}{{\mathcal {V}}}

\newcommand{\conefamilyc}{{\mathfrak{C}}}
\newcommand{\conefamilyd}{{\mathfrak{D}}}

\newcommand{\Hol}{\text{Hol}}
\newcommand{\Mer}{\text{Mer}}
\newcommand{\lin}{\text{lin}}
\nc{\Id}{\mathrm{Id}}
\nc{\ot}{\otimes}
\nc{\bt}{\boxtimes}
\nc{\id}{\mathrm{Id}}
\nc{\Hom}{\mathrm{Hom}}
\nc{\im}{{\mathfrak Im}}

\newcommand{\tddeux}[2]{\begin{picture}(12,5)(0,-1)
\put(3,0){\circle*{2}}
\put(3,0){\line(0,1){5}}
\put(3,5){\circle*{2}}
\put(3,-2){\tiny #1}
\put(3,4){\tiny #2}
\end{picture}}

\newcommand{\tdtroisun}[3]{\begin{picture}(20,12)(-5,-1)
\put(3,0){\circle*{2}}
\put(-0.65,0){$\vee$}
\put(6,7){\circle*{2}}
\put(0,7){\circle*{2}}
\put(5,-2){\tiny #1}
\put(6,5){\tiny #2}
\put(-5,8){\tiny #3}
\end{picture}}

%
%

\nc{\mge}{_{bu}\!\!\!\!{}}

\nc{\vep}{\varepsilon}

\def \e {{\epsilon}}
\nc{\syd}[1]{}
\newcommand{\cy}[1]{{\color{cyan}  #1}}
\newcommand{\sy}[1]{{\color{purple}  #1}}
\newcommand{\zb }[1]{{\color{blue}  #1}}
\newcommand{\zbd}[1]{}
\newcommand{\li}[1]{{\color{red} #1}}
\newcommand{\lit}[2]{\sout{\color{red}{#1}}{\color{red} #2}}
\newcommand{\lir}[1]{{\it\color{red} (Li: #1)}}

\nc{\prt}{P-}
\nc{\Prt}{P-}
\newcommand{\loc}{locality\xspace}
\newcommand{\Loc}{Locality\xspace}
\nc{\orth}{orthogonal\xspace}
\nc{\tloc}{L-}
\newcommand{\lset}{{\bf \loc SET }}
\newcommand{\set}{{\bf SET }}
\newcommand{\xat}{{X^{_\top 2}}}
\newcommand{\xbt}{{X^{_\top 3}}}
\newcommand{\xct}{{X^{_\top 4}}}

\newcommand{\gat}{{G^{_\top 2}}}
\newcommand{\gbt}{{G^{_\top 3}}}
\newcommand{\gct}{{G^{_\top 4}}}
\nc{\gnt}{{G^{_\top n}}}

\newcommand{\htwot}{{\calh ^{\ot_\top 2}}}
\newcommand{\hbt}{{calh ^{\ot_\top 3}}}
\newcommand{\hct}{{\calh ^{\ot_\top 4}}}

\title{An algebraic formulation of the locality principle in renormalisation}

\author{Pierre Clavier}
\address{Institute of Mathematics,
University of Potsdam,
D-14476 Potsdam, Germany}
\email{clavier@math.uni-potsdam.de}

\author{Li Guo}
\address{Department of Mathematics and Computer Science,
         Rutgers University,
         Newark, NJ 07102, USA}
\email{liguo@rutgers.edu}

\author{Sylvie Paycha}
\address{Institute of Mathematics,
University of Potsdam,
D-14469 Potsdam, Germany\\ On leave from the Universit\'e Clermont-Auvergne\\
Clermont-Ferrand, France}
\email{paycha@math.uni-potsdam.de}

\author{Bin Zhang}
\address{School of Mathematics, Yangtze Center of Mathematics,
Sichuan University, Chengdu, 610064, China}
\email{zhangbin@scu.edu.cn}

\date{\today}

\begin{abstract}  We study the mathematical structure underlying the concept of locality which lies at the heart of classical and   quantum field theory, and  develop a machinery used to  preserve locality  during the renormalisation procedure.
Viewing renormalisation in the framework of Connes and Kreimer  as the  algebraic Birkhoff factorisation of characters on a Hopf algebra with values in a Rota-Baxter algebra, we build  \loc variants of these algebraic structures, leading to a \loc variant of the algebraic Birkhoff  factorisation. This provides an algebraic formulation  of  the conservation of  locality   while renormalising.
 {As an application  in the context of the Euler-Maclaurin formula on cones, we renormalise the exponential generating function which sums over the lattice points in convex cones.}
For a suitable multivariate regularisation, renormalisation from the algebraic Birkhoff  factorisation  amounts to composition by a projection onto holomorphic multivariate functions.
\end{abstract}

\subjclass[2010]{08A55,16T99,81T15, 32A20, 	52B20 	}

\keywords{locality, renormalisation, algebraic Birkhoff factorisation, partial algebra, Hopf algebra, Rota-Baxter algebra, multivariate meromorphic functions, lattice cones}

\maketitle

{ \ }
\vspace{-2cm}

\tableofcontents

\allowdisplaybreaks

\section{Introduction}
\subsection{Locality in quantum field theory}
Locality is a widespread notion in mathematics and physics. In physics, the principle of locality states that an object is only directly influenced by its immediate surroundings. A theory which includes the principle of locality
is said to be a ``local theory". Various notions of locality are also used in  {analysis (local operators, local Dirichlet forms)},  geometry (geometric localisation, {locality in index theory}), algebra (localised rings) and number theory (local fields).

In classical field theory, for a classical action  ${\mathcal A}(f)={\mathcal B}(f, f)$ given by a bilinear form $ {{\mathcal B}}: C_c^\infty(U, \C^k)\times C_c^\infty(U, \C^k) \to \C$ on an open subset $U\subset \R^n$, the  locality principle  translates to

\begin{quote}  for any $f_1, f_2\in C_c^\infty(U, \C^k)$ if
  ${\text{ Supp}}(  f_1)\cap  {\text{ Supp}}(  f_2)=\emptyset$, then  ${{\mathcal B}}(f_1, f_2)=0$.
\end{quote}
We  interpret the binary relation  {$\top$ defined by}
\begin{equation}f_1\top f_2\Longleftrightarrow {\text{ Supp}}(  f_1)\cap  {\text{ Supp}}(  f_2)=\emptyset
  \mlabel{eq:topf}
\end{equation}
  on   pairs $\{f_1, f_2\}$ as an independence relation    on $C_c^\infty(U, \C^k)$.

In quantum field theory, the locality  principle governs the construction    of consistent subtraction (of divergences) algorithms that preserve locality during the renormalisation process. Subtracting divergences may be interpreted as resulting from the addition to the effective action of new properly chosen terms (known as counterterms) that are local polynomials in the fields and their derivatives. A systematic algorithm to subtract divergent momentum space integrals while preserving the fundamental postulates of relativistic quantum field theory including locality was proposed by    Bogoliubov, Parasiuk, Hepp and Zimmermann (abbreviated BPHZ renormalisation, and based on the so-called forest formula) \cite{BP,H,Z}. More recently,  Connes and   Kreimer \cite{CK} gave an interpretation of this algorithm by means of a coproduct which enables to build -- using  dimensional regularisation -- a renormalised map via its algebraic Birkhoff factorisation, regarded as an algebra homomorphism from the Hopf algebra of Feynman graphs to the Rota-Baxter algebra of meromorphic functions in one variable. Alternatively, following Speer \cite{S}, one can use analytic regularisation, which  gives rise to a map on graphs with values in multivariate meromorphic functions. \footnote{In the Epstein-Glaser formalism,  an analytic regularisation \`a la Speer yields   Feynman amplitudes  obeying  amongst other axioms, a factorisation condition reflecting the locality principle \cite[Theorem 10.1]{D}.} In that case, locality is reflected in the fact that this map preserves locality (what we call a \loc map in Definition \ref{defn:localmap}).

Separation of supports, which in QFT reflects independence of events,  also arises in the early algebraic study of locality, in terms of locality ideals, initiated by the work of H.-J.~Borchers~\mcite{BoH1,Yn}. Here a locality ideal is defined to be the two-sided ideal
generated by commutators of test functions with space-like separated supports. Its importance comes from the fact that quantum fields satisfying the requirement of local commutativity can be regarded as Hilbert space representations of the tensor algebra annihilating the locality ideal.
See~\mcite{BoR,BDF} for the recent progresses initiated by R.~E. Borcherds.

\subsection{Locality in algebraic Birkhoff factorisation}

In this paper we take an algebraic approach to investigate how locality is preserved in the renormalisation process, and choose to work in the framework of algebraic Birkhoff factorisation \`a la Connes and Kreimer~\mcite{CK}.
Our starting point is to view locality  as a symmetric binary relation comprising all pairs of independent events as in Eq.~(\mref{eq:topf}). Our main task is to explore the structures which are compatible with and preserve the locality, thus providing a mathematical formulation of the locality principle.

To make our point more precise, let us briefly recall the algebraic Birkhoff factorisation in the approach of Connes and Kreimer.

\begin{theorem} {\bf (Algebraic Birkhoff factorisation, Hopf algebra version)}
Let $H$ be a connected Hopf algebra and let $(A,P)$ be a commutative Rota-Baxter algebra of weight $-1$ with an idempotent Rota-Baxter operator $P$. Any algebra homomorphism $\phi: H\to A$ factors uniquely as the convolution product
\begin{equation}
\phi =\phi_-^{\star (-1)}\star \phi_+
 \mlabel{eq:abfhopf}
 \end{equation}
of algebra homomorphisms $\phi_-:H\to K+P(A)$ and $\phi_+:H\to K+(\Id-P)(A)$.
\mlabel{thm:abfhopf}
\end{theorem}

As an instance of physics applications, $H$ is the Connes-Kreimer Hopf algebra of Feynman diagrams, $A$ is the Rota-Baxter algebra $\C[\vep^{-1},\vep]]$ of Laurent series and $\phi$ is the regularisation map sending a Feynman diagram to the (dimensional) regularisation of the corresponding Feynman integral with Laurent expansion in $A$.

We can reformulate the theorem as follows: for a multiplicative regularisation map $\phi$, the  {renormalised map} $\phi_+$ is also multiplicative. Thus renormalisation preserves multiplicativity, a property which is the driving thread underlying the Hopf
algebra method introduced by Connes and Kreimer. Also, it is essential in the applications of the algebraic Birkhoff factorisation in mathematics, specifically when
renormalising multiple zeta values~\mcite{EMS,GZ,MP} while preserving their shuffle product and quasi-shuffle product.

In a recent study~\mcite{GPZ2} of the renormalisation of conical zeta values and Euler-Maclaurin formula on lattice cones, the algebraic Birkhoff factorisation was generalised to weaken the Hopf algebra condition of $H$ to a connected coalgebra as well as the Rota-Baxter algebra condition of $A$ to allow for an algebra with a decomposition $A=A_1\oplus A_2$ where only $A_1$ is a required to be a subalgebra of $A$.

\begin{theorem} {\bf (Algebraic Birkhoff factorisation, coalgebra version)} \cite[Theorem~4.4]{GPZ2}
Let $H$ be a connected coalgebra with coaugmentation $H_0=K\,J$ and let $A$ be a commutative (unital) algebra with an idempotent linear operator $P$ on $A$ such that $\ker P$ is a subalgebra of $A$. Any linear map $\phi: H\to A$ with $\phi(J)=1_A$ factors uniquely as the convolution product
\begin{equation}
\phi =\phi_-^{\star (-1)}\star \phi_+
 \mlabel{eq:abfcoalg}
 \end{equation}
of linear maps $\phi_-:H\to K+P(A)$ and $\phi_+:H\to K+(\Id-P)(A)$ with $\phi_\pm(J)=1_A$.
\mlabel{thm:abfcoalg}
\end{theorem}

As our motivation and first application of this generalised algebraic Birkhoff factorisation, $H$ is taken to be the connected coalgebra $\Q \bfc$, the linear span over $\Q$ of the set $\bfc$ of lattice cones, with the transverse coproduct, $A$ is the algebra $\calm _\Q$ of multivariant meromorphic functions with linear poles and rational coefficients, equipped with the direct sum $\calm _\Q =\calm_{\Q +}\oplus  \calm_{\Q -}$  where $\calm_{\Q +}$ is the algebra of holomorphic functions and $\calm_{\Q -}$ is the space of polar germs. See Section~\mref{ss:exam} for details. The linear map $\phi$ is
 \begin{equation}
 S: \Q \bfc \to \calm _\Q\quad \text{with} \quad S(C, \Lambda _C)( z) = \sum_{n\in C^o \cap \Lambda _C} e^ {\langle n, z\rangle}
 \mlabel{eq:expsum}
 \end{equation}
which corresponds to the exponential generating function summing over the lattice points in a lattice cone $(C, \Lambda _C)$.

Then a natural locality issue is whenever $S(C, \Lambda _C)$ and $S(D, \Lambda _D)$ are orthogonal under the locality relation on $\calm_\Q$ induced from the $\Q$-Euclidean space, to ask whether $S_+(C, \Lambda _C)$ and $S_+(D, \Lambda _D)$  also orthogonal. Another version of the question is, if $(C, \Lambda _C)$ and $(D, \Lambda _D)$ are orthogonal, whether $S_+(C, \Lambda _C)$ and $S_+(D, \Lambda _D)$  are orthogonal.
We reformulate this first question   {relative to} the locality principle in renormalisation in more general terms as  follows:
\begin{problem}
{\bf (Locality Conservation Principle in Renormalisation)} Consider a connected coalgebra $H$ and a commutative algebra $A$ with a linear map $P$ as in Theorem~\mref{thm:abfcoalg}. Let $\mtop_H\subseteq H\times H$ and $\mtop_A\subseteq A\times A$
be relations on $H$ and $A$ respectively. Let $\phi: H\to A$ be a linear map  compatible with the relations in the sense that $(\phi\times \phi)(\mtop_H)\subseteq \mtop_A$. Determine the conditions under which the  {renormalised map} $\phi_+$ is also compatible with the relations.
\mlabel{pr:lpcoalg}
\end{problem}

The first main goal of this paper is to provide a solution to this problem as a consequence of the locality generalisation (Theorem~\mref{thm:abflcoalg}) of the algebraic Birkhoff factorisation in Theorem~\mref{thm:abfcoalg}.

Note that the lack of multiplication in the coalgebra $H$ means that $\phi$ and $\phi_+$ are only linear maps, leaving out the other algebraic structures. However, as we will see next, it is precisely the interaction of the binary
relations with all the existing and potential algebraic structures involved, that makes the locality principle work. This interaction of binary relations with algebraic structures leads us to partially defined binary operations, which we dubbed \loc structures, throughout the whole hierarchy of algebraic structures from \loc set up to \loc algebra and \loc coalgebra, then further to \loc Rota-Baxter algebra and \loc Hopf algebra.

\subsection{Locality and partially defined operations}

In mathematics one often encounters multiplications which are meaningful only partially even if they might be everywhere defined. This phenomenon was long known in number theory where certain functions are multiplicative only for coprime positive integers, for instance Euler's totient function $\phi(n)$ counting the number of integers modulo $n\geq 1$ which are relatively prime to $n$ and the Ramanujan tau function $\tau(n)$ in modular forms. In fact, such phenomena have become so prevalent that such a restricted multiplicative function is simply called a multiplicative function in number theory~\mcite{Ap}. Such operations with restrictions can be viewed in the general framework of partial algebras in universal algebra~\mcite{Gr} (see Footnote~\ref{ft:part}).

An example of special importance for us is that of lattice cones. Even though the product given by the Minkowski sum  is defined for any two convex cones, and can be extended to any two lattice cones, compatibilities with either the coalgebra structure or the regularisation maps $\phi:H\to A$, such as the exponential sum, can be expected only  when the cones are orthogonal. See later sections for details on this (Propositions~\mref{prop:cones} and \mref{prop:conebialg}) and other instances.

This leads to another natural question:  whether, in the absence of a fully defined multiplicativity that is preserved by renormalisation, as in the classical algebraic Birkhoff factorisation in Theorem~\mref{thm:abfhopf}, one can hope for a partially defined multiplication preserved by renormalisation. So we propose the following

\begin{problem}
{\bf (\Loc Product Conservation Principle in Renormalisation)} Consider a connected bialgebra $H$
and a unital commutative algebra $A$ with a linear map $P$ as in Theorem~\mref{thm:abfhopf}. Let $\mtop_H\subseteq H\times H$ and $\mtop_A\subseteq A\times A$
be relations on $H$ and $A$ respectively for which a partial multiplication $m_H:\mtop_H \to H$ and $m_A:\mtop_A\to A$ are defined. Let $\phi: H\to A$ be partially multiplicative in the sense that $\phi(m_H(u,v))=m_A(\phi(u),\phi(v))$ for $(u,v)\in \mtop _H$. Determine a condition under which the  {renormalised map} $\phi_+$ is also partially multiplicative.
\mlabel{pr:lphopf}
\end{problem}

At this point it is worthwhile to observe the mutual effects of the interplay between locality relations and the partial algebraic structures mentioned above. In one direction it allows us to pass the locality of $\phi$ represented by the relations onto the corresponding  {renormalised map} $\phi_+$, thus giving a solution of the  Locality {Conservation} Principle in Problem~\mref{pr:lpcoalg}. This is achieved in Theorem~\mref{thm:abflcoalg}.
In the other direction this interplay allows us to transmit the partial multiplicativity of $\phi$ onto $\phi_+$, thus giving a solution of the \Loc Product Conservation Principle in Problem~\mref{pr:lphopf}. This is achieved in Theorem~\mref{thm:abflhopf}.

Of particular interest to us is the exponential generating sum $S:\Q \bfc \to \calm_\Q$ mentioned above. In spite of the fact that the space $\Q \bfc$ of lattice cones is a genuine algebra when equipped with the extended Minkowski product,
the product is not compatible with the transverse coproduct, so we do not have a bialgebra. Likewise, in the decomposition $\calm _\Q=\calm_{\Q,+}\oplus  \calm _{\Q,-}$, even though the summand $\calm _{\Q,+}$, the space of holomorphic germs,
is a subalgebra, the summand $ \calm _{\Q,-}$, the space of polar germs, is not. Hence this decomposition does not give a Rota-Baxter algebra.  Furthermore, the linear map $S$ does not send the Minkowski product on $\Q \bfc$ to the function multiplication
in $\calm _\Q$. However if one considers only orthogonal pairs of lattice cones and suitable orthogonal relation of meromorphic germs,  all these structures can be recovered in the form of \loc structures. In fact $\Q \bfc$ is not only a \loc bialgebra, it is a  connected \loc Hopf algebra. Moreover,  $ \calm _{\Q,-}$ is not only a \loc subalgebra, it is a \loc ideal, showing that the projection
$\pi_+:\calm _\Q\to \calm _{\Q,+}$  onto $\calm _{\Q,+}$ along $\calm_{\Q -}$ is a \loc algebra homomorphism. Consequently,  the full algebraic Birkhoff factorisation can be recovered on the locality level, from which it then follows that the  {renormalised map} $\phi_+$ is a
\loc algebra homomorphism. Further the   \loc ideal  property of $\calm _{\Q,-}$ implies that its convolution inverse $\phi_+^{\star -1}$ is $\phi$ composed with the projection of $ \pi_+$ (see Eq.~\eqref{eq:BHlocalpi}). This applies whenever the regularisation map $\phi$ is a
\loc algebra homomorphism and takes values in $\calm _\Q$. In this situation a recursive formula for $\phi_-$  in terms of the projections $\pi_+$ and $\Id- \pi_+$ is  given in \eqref{eq:phi2},  which is reminiscent of the forest formula of the renormalisation of Feynman graphs in the context of Quantum Field Theory.

\subsection{Outline of the paper}

We next give a summary of the locality construction as an outline of the paper.

We consider vector spaces $H$ and $A$, each equipped with a suitable symmetric binary relation, and a linear map $\phi:H\to A$ preserving the relations. In order to pass this property of $\phi$ onto the corresponding  {renormalised map} $\phi_+$ via the algebraic Birkhoff factorisation for a suitably enriched $H$, $A$ and $\phi$, we need to make the relations compatible with all the algebraic structures involved in the algebraic Birkhoff factorisation, including a Hopf algebra or bialgebra structure on $H$, a Rota-Baxter algebra or algebra structure on $A$, and the corresponding structures on $\phi$.

Throughout the paper we use the space of convex lattice cones, the space of meromorphic functions and the exponential generating sum in Eq.~(\mref{eq:expsum}) between the two spaces as the primary examples of the algebraic constructions, and of our main theorems on locality.  {Further applications will be given in   subsequent work}, {such as locality of branched zeta values~\mcite{CGPZ}.}   {The locality relation relates to Weinstein's ``congeniality" condition \cite{LW,W} in a selective category, which in turn is motivated by the drive to build a quantising functor from the category of  canonical relations between symplectic manifolds to a category of quantum morphisms. We hope to explore these  relations in   future work. }

Thus we begin in Section~\mref{sec:lset} with the general concepts of a \loc set and \loc map, emphasising examples on lattice cones and meromorphic functions, while mentioning several other examples in passing. In Section~\mref{sec:lalgebra}, we  equip a \loc set with a compatible associative multiplication to give a \loc semigroup, \loc monoid and \loc group. Then through the intermediate structure of a \loc vector space, we obtain a \loc algebra and further a \loc Rota-Baxter algebra. In Section~\mref{sec:lcoalg}, we consider the coalgebraic aspect of the construction which begins with the preliminary but subtle notion of \loc tensor product. With it, we introduce the concepts of \loc coalgebra, the convolution product for maps from a \loc coalgebra to a \loc algebra. At this point we can give our first main result, Theorem~\mref{thm:abflcoalg}, addressing the Locality  Conservation Principle in Problem~\mref{pr:lpcoalg}. Finally in Section~\mref{sec:lhopf}, we bring the \loc algebra and \loc coalgebra together to form a \loc bialgebra and further a \loc Hopf algebra under an extra connectedness condition. Then we prove our second main result, Theorem~\mref{thm:abflhopf}, addressing the  \Loc Product Conservation Principle in Problem~\mref{pr:lphopf}. Both results are applied to the example of the exponential generating sum $S:\Q \bfc\to \calm _\Q$, showing that the orthogonality property and the partial multiplicativity on orthogonal pairs of convex lattice cones are indeed preserved by the renormalised version of $S$.

\smallskip

\noindent
{\bf Notations.}
Unless otherwise specified, all algebras are taken to be unitary commutative over a field $K$, and linear maps and tensor products are taken over $K$. A nonunitary algebra means one which does not necessarily have a unit. The same applies to coalgebras.

\section{Locality for sets and maps}
\mlabel{sec:lset}

\subsection{Concepts of \loc sets and \loc maps}
\mlabel{ss:conc}

We begin with a set with an independent relation.
 \begin{defn}
 		\begin{enumerate}
\item An {\bf \loc set} is a  couple $(X, \top)$ where $X$ is a set and $ \top \subseteq X\times X$ is a symmetric relation on $X$, referred to as the
{\bf \loc relation} (or {\bf independence relation}) of the \loc set. So for $x_1, x_2\in X$, denote $x_1\top x_2$ if $(x_1,x_2)\in \top$. When the underlying set $X$ needs to be emphasised, we use the notation $X\times_\top X$ or $\mtop_X$ for $\top$.
 \item  For any subset $U$ of a \loc set $(X,\top)$, let
 			\begin{equation*}
 			U^\top:=\{x\in X\,|\, (U,x)\subseteq X\times_\top X\}
 \end{equation*}
be the {\bf  polar subset} of $U$.
\end{enumerate}
\mlabel{defn:independence} 	
\end{defn}
 	
\begin{rk}
\begin{enumerate}
\item
Thus a \loc set is simply a set with a binary symmetric relation: we use the term \loc set   to be consistent with the  derived terminology to be introduced later for various algebraic structures built on top of the \loc  set.
\item The binary relation $\top$ plays two roles in our study which are related and yet complementary. On the one hand, it serves as a condition under which two elements are related in various ways (independent, orthogonal, etc.), hence
the symmetry requirement. As noted in the introduction, the \loc relation is intended to encode the notion of independence of events in physics (thus the alternative
 name independence relation); on the other hand, it assigns the subset of $X\times X$ as the domain for partial binary operations in the context of universal algebra~\mcite{Gr}. Strictly speaking, the symmetric condition is not required for the latter purpose even though in most of our applications, the algebras are commutative and hence the domain is symmetric. As the two roles are so closely related, we will only deal with symmetric relations unless otherwise needed.
\end{enumerate}
\end{rk}

From a \loc set one can derive other \loc sets as follows.
\begin{lem}
Let $(X,\top)$ be a \loc set.
\begin{enumerate}
\item For a subset $X'$ of $X$, denote $\top':=(X'\times X')\cap \top$. Then the pair $(X',\top')$ is a \loc set, called a {\bf sub-\loc set} of $(X,\top)$;
\item For subsets $A$ and $B$ of $(X,\top )$, denote $A\top^\calp B$, or simply $A\top B$ should the context be clear, if
$ A\times B\subseteq \top$.
Then $\top^\calp$ equips the power set $\calp(X)$ of $X$ with a \loc set structure;
\item Combining the above two items, any subset $Y$ of $\mathcal{P}(X)$ with the restriction of $\top^\calp$ defines a \loc set $(Y,(Y\times Y)\cap \top^\calp)$.
\end{enumerate}
\mlabel{lem:derive}
\end{lem}

\aside{
Here are further straightforward properties of \loc sets.
\begin{prop}
Let  $\left(X, \,\top\, \right)$ be a \loc set and let $U, V\subseteq X$.
	\begin{enumerate}
			\item Due to the symmetry of the relation we have
			\begin{equation}\mlabel{eq:Utop}
V\subset U^\top\Longleftrightarrow U\times V \subseteq \top \Longleftrightarrow U\subset V^\top.
\end{equation}
\item For any subset $U\subset X$,
$$ \left(U^\top\right)^\top\supset U.$$ However, in general $\top$ is not involutive:
$ \left(U^\top\right)^\top\neq U$.
\item   If $X^{\top}=\{e\}$ we have
\begin{equation*}
e^\top:=\{x\in X\,|\, x\top e\}=X\quad \text{and hence}\quad  \left(X^\top\right)^\top=X\quad \text{and}\quad \left(e^\top\right)^\top=\{e\}.
 			\end{equation*}
 		\end{enumerate}
\end{prop}
} 	

 {As a very simple yet fundamental example which underlies the more sophisticated examples to be discussed below, the orthogonality relation $\perp^Q$ between vectors or subsets in an Euclidean vector space $(V,Q)$ equips $V$ or the set of subsets of $V$ with the structure of a \loc set. }
	
\begin{defn}
A {\bf \loc map} from a \loc set $\left(X,\mtop_X\right)$ to a \loc set $ (Y, \mtop_Y)$ is a map $\phi:X\to Y$ such that $(\phi\times \phi)(\mtop_X)\subseteq \mtop_Y$. More generally, maps $\phi,\psi:\left( X,\mtop_X\right)\to \left(Y, \mtop_Y\right)$ are called {\bf independent} and denoted $\phi\top \psi$ if
$(\phi\times \psi)(\mtop_X) \subseteq \mtop_Y$.
To be specific, if $\phi\neq \psi$, $\phi$ is called independent of $\psi$.
\mlabel{defn:localmap}
\end{defn}
	 {
\begin{ex}Any orthogonal map between two Euclidean vector spaces $(V_i,Q_i), i=1,2$, is a \loc map between the locality spaces $(V_i,\perp^{Q_i})$.
\end{ex}}
\begin{rk}
\begin{enumerate}
\item
The identity map  on a \loc set $\left( X,\top \right)$ is trivially a \loc map. Also, the composition of two \loc maps is still a \loc map. Thus \loc sets and \loc maps form a category.
\item
	 A map independent of the identity is a \loc map. Indeed let $(\Omega,\mtop_{\Omega})$ be a \loc set and $\phi:\Omega\mapsto\Omega$ be a
	 map such that $\phi\top {\rm Id}_{\Omega}$. Then for any $(x,y)\in\mtop_{\Omega}$ we have
	 \begin{equation*}
	  \phi(x)\mtop_{\Omega}{\rm Id}_{\Omega}(y)\Rightarrow y\mtop_{\Omega}\phi(x)\Rightarrow\phi(y)\mtop_{\Omega}{\rm Id}_{\Omega}(\phi(x))\Rightarrow\phi(x)\mtop_{\Omega}\phi(y).
	 \end{equation*}
Note that here we need the symmetric condition.
\end{enumerate}
\mlabel{rk:mut-ind-loc}
\end{rk}

\subsection{Examples of \loc sets and \loc maps}
\mlabel{ss:exam}

\subsubsection{Convex lattice cones and meromorphic functions}

 We now give some background for the main examples which serve as both the motivation and prototype of the theoretical structures in this paper. See~\mcite{GPZ2,GPZ4} for details.

Our first example of \loc sets is given by convex polyhedral lattice cones. Consider the filtered rational Euclidean lattice space
$$\Big(\R ^\infty=\bigcup_{\geq 1} \R ^k, \Z ^\infty =\bigcup_{\geq 1} \Z ^k, Q=(Q_k(\cdot, \cdot))_{k\geq 1}\Big),$$
where
$$ Q_k(\cdot,\cdot): \R ^k\ot \R ^k \to \R, \quad k\geq 1,$$
is an inner product in $\R ^k$ such that $Q_{k+1}|_{\R ^k\ot \R ^k}=Q_k$ and $Q_k(\Z ^k\otimes \Z ^k)\subset \Q$,
A lattice cone is a pair $(C, \Lambda _C)$ where $C$ is a polyhedral cone in some $\R ^k$, that is,
$$ C=\langle u_1,\cdots,u_m\rangle:=\Big\{\sum_{i=1}^m c_i u_i\,\Big|\, c_i\in \R_{\geq 0}, 1\leq i\leq m\Big\}
$$
for some $u_1,\cdots,u_m \in \Q ^k$, and $\Lambda _C$ is a lattice in the linear subspace spanned by $C$.
Let $\bfcf _k$ be the set of lattice cones in $\R ^k$ and
 $$\bfcf =\bigcup_{k\geq 1} \bfcf_k
 $$
be the set of lattice cones in $(\R ^\infty, \Z ^\infty)$. Let $\Q \bfc _k $ and $\Q \bfc $ be the linear spans of $\bfcf _k $ and $\bfcf$ over $\Q$.

In $(\R ^\infty, \Z ^\infty , Q)$, we write $\perp^Q$ for the corresponding orthogonality relation.
\begin{defn}
We call two lattice cones $(C, \Lambda _C)$ and $(D, \Lambda _D)$ {\bf \orth} (with respect to $Q$), if $Q(u,v)=0$ for all $u\in \Lambda_C, v\in \Lambda_D.$
Then we write $(C, \Lambda _C)\perp^Q (D, \Lambda _D)$.
\mlabel{de:perpcone}
\end{defn}

Multivariate meromorphic functions provide another fundamental motivation for the  forthcoming algebraic setup. Again in $(\R ^\infty, \Z ^\infty , Q)$, let $\calm_\Q((\R ^k)^* \ot \C)$ be the space of meromorphic germs at $0$ with linear poles and rational coefficients~\mcite{GPZ2,GPZ4} and let
\begin{equation}
\calm _\Q :=\bigcup_{k\geq 1} \calm_\Q((\R ^k)^*\ot \C).
\mlabel{eq:mero}
\end{equation}
An element of $\calm _\Q$ can be written as a sum of a holomorphic germ and elements the form
\begin{equation}
\frac{h(\ell_1,\cdots,\ell_m)}{L_1^{s_1}\cdots L_n^{s_n}}, \quad s_1, \cdots, s_n\in \Z_{>0},
 \mlabel{eq:polar}
 \end{equation}  where $h$ is a holomorphic germ with rational coefficients in linear forms $\ell_1,\cdots,\ell_m\in \Q ^k$,  and $L_1,\cdots, L_n$ are linearly independent linear forms in $\Q ^k$, $\ell_i\perp^QL_j$ for all $i\in \{1,\cdots, m\} \tforall j\in \{1, \cdots, n\}$, which is called a {\bf polar germ.}

\begin{defn}  Two meromorphic germs with rational coefficients $f$ and $f'$  are {\bf $Q$-\orth } which we  denote by $f\perp^Q f'$ if there exist linear functions $L_1, \cdots , L_m \in \Q ^k$ and $L _1', \cdots , L' _n\in \Q ^k$ satisfying $Q(L_i, L' _j)=0$ for $i=1, \cdots, m$, $j=1, \cdots, n$,
and meromorphic germs $g\in \calm _\Q (\R ^m\otimes \C) $ and $g'\in\calm _\Q (\R ^n\otimes \C)$, such that $f=g(L_1, \cdots, L_m)$, $f'=g'(L_1', \cdots, L_n')$.
Let $(\calm_\Q, \perp^Q)$ denote the resulting \loc set.
\mlabel{de:meroindep}
\end{defn}

We next give examples of \loc maps. Let $(C,\Lambda_C)$ be a strongly convex lattice cone in $\R ^k$ with interior $C^o$. For $z$ in the   dual  cone
$$C^-:=\{z \in (\R ^k)^*\,|\,  \langle x,z\rangle <0, \forall x\in C\},$$
we define its {\bf exponential generating function} to be the sum
$$ S(C,\Lambda _C)(z):= \sum_{n\in C^o\cap \Lambda _C} e^{\langle n,z\rangle }.$$
We also define its {\bf exponential integral} $I(C,\Lambda _C)$ to be the integral
$$ I(C, \Lambda _C)(z):= \int _C e^{\langle x,z\rangle }d\Lambda_x,$$
where $d\Lambda _x$ is the volume form induced by generators of $\Lambda _C$ such that the polytope generated by a basis of $\Lambda _C$ has volume 1.

These assignments extend by subdivisions to maps:
\begin{equation}
S, I: \bfcf \to \calm _\Q.
\mlabel{eq:conemeromap}
\end{equation}

\begin{prop}  \label{pp:conemeromap}
For lattice cones $(C, \Lambda _C)$ and $(D, \Lambda _D)$, if $(C, \Lambda _C)\perp ^Q (D, \Lambda _D)$, then $S(C, \Lambda _C)(z)\perp  ^Q S(D, \Lambda _D)(z)$ and $I(C, \Lambda _C)(z)\perp  ^Q I(D, \Lambda _D)(z)$ in the sense of
Definition~\mref{de:meroindep}, that is, the exponential integral and exponential generating function maps $I$ and $S$ are  \loc maps.
\end{prop}

\begin{proof} For any subdivision $\{(C_i, \Lambda _{C_i})\}$ of $(C, \Lambda _C)$, since $\Lambda _{C_i}=\Lambda _{C}$, we know $(C_i, \Lambda _{C_i})\perp ^Q (D, \Lambda _D)$. Because any lattice cone can be subdivided into smooth lattice cones, we can reduce the proof to smooth lattice cones.

For a smooth lattice cone $(C, \Lambda _C)=(\langle u_1,\cdots, u_n\rangle, \sum_{i=1}^n \Z u_i)$, we have
$$ S(C,\Lambda _C)(z)=\prod_{i=1}^n \frac{ e^{\langle u_i, z\rangle }}{1-e^{\langle u_i, z\rangle }}; \quad
  I(C,\Lambda _C)(z)= \prod_{i=1}^k \frac{1}{ \langle u_i, z\rangle }.$$
So $S(C,\Lambda _C)(z)$ and $S(D,\Lambda _D)(z)$ (resp. $I(C,\Lambda _C)(z)$ and $I(D,\Lambda _D)(z)$) are meromorphic germs in perpendicular linear functions. Therefore $S(C,\Lambda _C)(z) \perp ^Q S(D,\Lambda _D)(z)$ and $I(C,\Lambda _C)(z) \perp ^Q I(D,\Lambda _D)(z)$.
\end{proof}

\subsubsection{Other examples}
\mlabel{ss:othex}
There are many other examples of \loc sets. To save space, we only briefly list some of them and refer the reader to the references for further details.

A large number of examples come from disjointness of subsets noted in Lemma~\mref{lem:derive}.

 \begin{enumerate}
\item
On the one hand, locality  structures can be built on functions or distributions by  requiring the disjointness of   (adequately chosen) supports of  such maps, such as their ordinary supports, their singular supports  or  wavefront sets~(see \mcite{BDH} and references therein).
\item
On the other hand, locality structures on  maps can be built by requiring disjointness of their image sets. This is the case for  the decorating maps from the vertices of graphs or trees to a decorating set, which yields \loc sets of labelled graphs and trees~\mcite{F}.
\item
When the set $S$ is equipped with a linear structure, we can replace disjointness by trivial intersections, transversality or linear independence. If $S$ further has an inner product, then the relation can be chosen to be the one of orthogonality.
\end{enumerate}

Further examples of \loc sets include independence of events in probability and coprimeness of natural numbers as discussed in the introduction.

\section{Building up locality from semigroups to Rota-Baxter algebras}
\mlabel{sec:lalgebra}
In this section, we equip a \loc set with various algebraic structures, from that of a semigroup to that of a Rota-Baxter algebra.

\subsection {\Loc semigroups}
\mlabel{ss:lsg}

For a \loc set $(X,T)$ and an integer $k\geq 2$, denote
\begin{equation}
X^{_\top k}: = \underbrace{X \times_\top \cdots \times_\top X}_{k \text{ factors}}: = \left\{ (x_1,\cdots,x_k)\in X^k\,\left|\, (x_i,x_j)\in \top, 1\leq i\neq j\leq k\right.\right\}.
\mlabel{eq:lprod}
\end{equation}

\begin{defn} \mlabel{defn:lsg}
\begin{enumerate}
\item
An {\bf \loc semigroup}\footnote{\label{ft:part}As a special case of partial algebras~\mcite{Gr}, the terminology ``partial semigroup" is used for a set equipped with a partial associative product defined only for certain pairs of elements in the set. The condition for a \loc semigroup is more restrictive than that of a partial semigroup in that the former requires that the pairs {for which  the partial product is defined stem  from} a symmetric relation and that
the partial product should be compatible with the \loc relation in the sense of Eq.~(\mref{eq:semigrouploc}).} is a \loc set $(G,\top)$ together with a product law defined on $\top$:
$$ m_G: G\times_\top G\longrightarrow  G
$$
for which the product is compatible with the \loc relation on $G$, namely
\begin{equation}\tforall U\subseteq G, \quad  m_G((U^\top\times U^\top)\cap\top)\subset U^\top
\mlabel{eq:semigrouploc}\end{equation}
and such that the following {\bf \loc associativity property} holds:
\begin{equation}\mlabel{eq:asso}
(x\cdot y) \cdot z = x\cdot (y\cdot z) \text{ for all }(x,y,z)\in G\times_\top G\times_\top G. 	
\end{equation}
Note that, because of the condition \eqref{eq:semigrouploc}, both sides of Eq.~(\mref{eq:asso}) are well-defined for any triple in the given subset.
\mlabel{it:lsg}
\item
An \loc semigroup is {\bf commutative} if $m_G(x,y)=m_G(y,x)$ for $(x,y)\in \top$, noting that both sides of the equations are defined since $\top$ is symmetric.
\item
An {\bf \loc   monoid} is a \loc   semigroup $(G,\top, m_G)$ together with a {\bf unit element} $1_G\in G$ given by the defining property
\[\{1_G\}^\top=G\quad \text{ and }\quad m_G(x, 1_G)= m_G(1_G,x)=x\quad \tforall  x\in G.\]
We denote the \loc  monoid by $(G,\top,m_G, 1_G)$.
\mlabel{defn:partial monoid}
\item An {\bf \loc group} is a \loc monoid $(G,\top, m_G,1_G)$ equipped with a morphism $\iota: G\to G$ of \loc sets, called the {\bf inverse map}, such that  $(\iota (g),g)\in \top$ and $m_G(\iota(g),g)=m_G(g,\iota(g))=1_G$ for any $g\in G$.
\item
A {\bf sub-\loc semigroup} of a \loc semigroup $(G,\top,m_G)$ is a \loc semigroup $(G',\top',m_{G'})$ with $G'\subseteq G$, $\top'=(G'\times G')\cap \top$ and $m_{G'}=m_G|_{\top'}$, that is, for $x, y\in G'$ and $(x, y)\in \top$, $m_G(x,y)$ is in $G'$.
A {\bf sub-\loc monoid} of a \loc monoid is a sub-\loc semigroup of the corresponding \loc semigroup which share the same unit. A {\bf sub-\loc group} of a \loc group is a sub-\loc monoid of the corresponding \loc monoid which
is also a \loc group.
\mlabel{it:lssg}
\end{enumerate}
\end{defn}
For notational convenience, we usually abbreviate $m_G(x,y)$ by $x\cdot y$ or simply $xy$.

\begin{rk}
One easily checks that on   a \loc   monoid $(G, \top,m_G, 1_G)$ if $(x_1,x_2, y_1, y_2)$ is in $\gct$, then $(x_1 x_2, y_1, y_2)$ and $(x_1,x_2, y_1 y_2)$ are in $\gbt$ and hence $(x_1x_2, y_1y_2)\in \top.$
\end{rk}

As a simple counter example of \loc semigroup, we have

\begin{coex}
The set $G$ of linear subspaces of $\R^2$ is a \loc set with respect to the following relation $\mtop_G$ on linear subspaces of $G$: $U, V\subseteq \R^2$ are called transverse if they intersect trivially, namely if
$U\cap V=\{0\}$. The set $G$ equipped with linear sums $+$  is  a monoid. But the corresponding $(G, \mtop_G, +)$ is not a \loc monoid. Indeed, for the standard basis $\{e_1,e_2\}$ of $\R^2$,  the subspaces $\R e_1$ and $\R e_2$  both  intersect $\R (e_1+e_2)$ trivially, but $\R e_1 +\R e_2$ does not.
\end{coex}

\begin {ex}
The \loc set $(\calm _\Q,\perp^Q)$ in Definition~\mref{de:meroindep}, equipped with the restricted multiplication $m_Q:\calm_\Q \times_{\perp^Q} \calm_\Q \mapsto\calm_\Q$,
  is a \loc monoid.
\end{ex}

\begin{defn} Let $\left(X,\mtop_X,\cdot_X  \right) $ and $\left(Y,\mtop_Y,\cdot_Y  \right) $ (resp. $\left(X,\mtop_X,\cdot_X, 1_X  \right) $ and $\left(Y,\mtop_Y,\cdot_Y, 1_Y  \right) $) be \loc semigroups (resp.   monoids). A map $\phi:X\longrightarrow Y$ is called a {\bf \loc semigroup} (resp. {\bf \loc monoid}) {\bf homomorphism}, if it
				\begin{enumerate}
					\item  is a \loc map;
					\item is {\bf \loc multiplicative}:
for $({ a},{ b})\in \mtop_X$ we have
$\phi({ a}\cdot_X{ b})= \phi({ a})\cdot_Y\phi({ b})$,
		\item (resp. preserves the unit  $\phi(1_X)=1_Y$.)
			\end{enumerate}
\mlabel{defn:Phipartialmult}
\end{defn}

\begin{ex}
Classical examples of \loc monoid homomorphisms are given by {\bf multiplicative functions} in number theory. Here a function $f: \Z_{\geq 1} \to \Z_{\geq 1}$ is multiplicative means $f(1)=1$ and $f(mn)=f(m)f(n)$ if $m$ and $n$ are coprime. This means precisely that $f$ is a \loc monoid homomorphism from the \loc monoid $(\Z_{\geq 1},\mtop_{\rm cop})$ where $\mtop_{\rm cop}$ is the coprime relation, to the \loc monoid $(\Z_{\geq 1},\mtop_{\rm full})$, where $\mtop_{\rm full}$ is the full relation $\Z_{\geq 1}\times \Z_{\geq 1}$.
\mlabel{ex:intprod}
\end{ex}

Now let us take a closer look at the set $\bfcf$ of lattice cones. For convex cones $C:=\langle u_1,\cdots,u_m\rangle$ and $D:=\langle v_1,\cdots,v_n\rangle$ spanned by $u_1,\cdots,u_m$ and $v_1,\cdots,v_n$ respectively, their {\bf Minkowski product} (usually called Minkowski sum) is the convex cone
\begin{equation*}
C\cdot D:=\langle u_1,\cdots,u_m,v_1,\cdots, v_n\rangle.
\end{equation*}
This product can be extended to a product in $\bfcf$:
\begin{equation}
(C, \Lambda _C)\cdot (D, \Lambda _D):=(C\cdot D, \Lambda _C+\Lambda _D),
\mlabel{eq:minkprod}
\end{equation}
where $\Lambda _C+\Lambda _D$ is the abelian group generated by $\Lambda _C$ and $\Lambda _D$ in $\Q ^\infty$.
This product endows a monoid structure on $\bfc$ with unit $(\{0\}, \{0\})$, which also restricts to a \loc monoid structure on $(\bfcf,\perp ^Q)$.

Even though $\bfcf$ and $\calm _\Q$ have their natural multiplications defined on the full spaces, the importance of the \loc structures on $\bfcf$ and $\calm _\Q$ becomes evident when studying the multiplicative property of the maps $I$ and $S$ from $\bfcf$ to $\calm _\Q$ introduced in Section~\mref{ss:exam}. Because of the idempotency $(C, \Lambda _C)\cdot (C, \Lambda _C)=(C, \Lambda _C)$ for $(C, \Lambda _C)\in \bfcf$, the multiplicativity $I((C, \Lambda _C)\cdot (C, \Lambda _C))(z)=I(C, \Lambda _C)(z)I(C, \Lambda _C)(z)$ or $S((C, \Lambda _C)\cdot (C, \Lambda _C))(z)=S(C, \Lambda _C)(z)S(C, \Lambda _C)(z)$ {\em does not} hold in general since that would force the integral or the sum to be $0$ or $1$, which can not be the case for example by taking $(C,\Lambda _C)=(\langle e_1\rangle, \Z e_1)$.
But the multiplicativity can be recovered in the context of \loc monoids, as follows.

\begin {prop}
If $(C, \Lambda _C)\perp^Q (D, \Lambda _D)$,
then
$$S((C, \Lambda _C)\cdot (D, \Lambda _D))=S(C, \Lambda _C)S(D, \Lambda _D), \quad I((C, \Lambda _C)\cdot (D, \Lambda _D))=I(C, \Lambda _C)I(D, \Lambda _D).$$
Thus the \loc maps $I$ and $S$ are \loc semigroup homomorphisms from $(\bfcf,\perp ^Q)$ to $(\calm _\Q,\perp ^Q)$.
\mlabel{prop:cones}
\end{prop}

\begin {proof} For $(C, \Lambda _C)\perp^Q (D, \Lambda _D)$, if both are smooth,
let $C:=\langle u_1,\cdots,u_m\rangle $ with $\Lambda _C=\oplus \Z u_i$, and $D:=\langle v_1,\cdots,v_n\rangle$ with $\Lambda _D=\oplus \Z v_j$. Since $Q(u_i,v_j)=0$, we have $\Lambda _C +\Lambda _D=\big(\oplus_i \Z u_i\big)\oplus \big(\oplus_j\Z v_j\big)$.
Thus  $(C, \Lambda _C)\cdot (D, \Lambda _D)$ is smooth. By a direct calculation, we obtain
$$S\left((C, \Lambda _C)\cdot (D, \Lambda _D)\right)=S(C, \Lambda _C)\,S(D, \Lambda _D), \quad I\left((C, \Lambda _C)\cdot (D, \Lambda _D)\right)=I(C, \Lambda _C\,I(D, \Lambda _D).$$
The general case follows by a subdivision of lattice cones into smooth lattice cones.
\end{proof}

Moreover, the \loc monoid $(\bfcf,\perp ^Q)$ has a natural grading which does not extend to the monoid $\bfcf$. For $n\geq 0$, let $\bfcf_n$ denote the subset of $\bfcf$ consisting of lattice cones of dimension $n$. Then for \orth  lattice cones $(C, \Lambda _C)$ and $(D, \Lambda _D)$, we have $\dim ((C, \Lambda _C)\cdot (D, \Lambda _D))=\dim (C, \Lambda _C) + \dim (D, \Lambda _D)$. Hence,
\begin{equation}
\bfcf=\sqcup_{n\geq 0} \bfcf_n, \quad m_{\bfcf}((\bfcf_m \times \bfcf_n)\cap \perp ^Q) \subseteq \bfcf_{m+n} \quad \tforall m, n\geq 0,
\mlabel{eq:conegrad}
\end{equation}
which we take to be the defining conditions for $\bfcf$ to be a {\bf graded \loc monoid}.

\subsection{\Loc vector spaces}
\mlabel{ss:lvecsp}
We now consider \loc relations on vector spaces.

\begin{defn}
An {\bf \loc vector space} is a vector space $V$ equipped with a \loc relation $\top$ which is compatible with the linear structure on $V$ in the sense that, for any  subset $X$ of $V$, $X^\top$ is a linear subspace of $V$.
\mlabel{defn:lvs}
\end{defn}

\begin{rk}
 For a \loc vector space $(V,\top)$, since $V^\top$ is a linear subspace of $V$, we have  $\{0\}\times V\subset \top,$ or equivalently  $0\in V^\top.$ Note that there is no \loc restrictions for the vector space structure (addition and scalar product) on $V$, that is, the addition and scalar product are everywhere defined.
\mlabel{rk:locvecsp}
\end{rk}

\begin{ex}
The vector space $\calm _\Q $ equipped with the relation $\perp^Q$ in Definition~\mref{de:meroindep} is a \loc vector space $\left( \calm _\Q,\perp^Q\right)$.
\mlabel{ex:merosp}
\end{ex}

\begin{rk}
Clearly, constant germs are \orth  of any germs, namely $\R {\subseteq} \calm^{\perp^Q}_\Q $. In fact, the converse is also true. Thus $\calm^{\perp^Q}_\Q=\R$.
\end{rk}

From a \loc set $(X, \top)$ we can build a \loc vector space  $(KX, \top)$ from the vector space generated by $X$ whose defining relation (denoted by the same symbol $\top$) is  the  linear extension of that on $X$. More precisely for $u, v\in   K X$, $(u,v)\in \top$ if the basis elements from $X$ appearing in $u$ are related via $\top$ to the basis elements appearing in $v$.
Thus
$$  KX\times_{\mtop_{  KX}}  KX=\bigcup_{U,V\subseteq X,(U,V)\subseteq \top}  KU\times   KV.$$

\begin{ex}
The \loc set $\bfcf$ of lattice cones with the orthogonal relation in Definition~\mref{de:perpcone} gives rise to the corresponding \loc vector space $\Q \bfcf$.
Likewise, the \loc set of labelled rooted trees described in Section~\mref{ss:othex} gives rise to the corresponding \loc  vector space generated by the set. \mlabel{ex:conesp}
\end{ex}

\begin{defn}  \mlabel{defn:locallmap}
Let $\left(U,\mtop_U\right)$ and $ (V, \mtop_V)$ be \loc vector spaces, a linear map $\phi:\left( U,\mtop_U\right)\to \left(V, \mtop_V\right)$ is called a {\bf \loc linear map} if it is a \loc map.
\end{defn}

\begin{ex}The \loc maps given by the exponential integral $I:\bfcf \to \calm_\Q$ and the exponential generating sum $S:\bfcf \to \calm_\Q$ in Proposition~\mref{pp:conemeromap} extend to \loc linear maps
from $\Q \bfcf$ to $\calm_\Q$.
\mlabel{ex:conemeroli}
\end{ex}

Here are further useful properties of \loc linear maps.
\begin{prop}

Let $(U,\mtop_U), (V, \mtop_V)$ be \loc vector spaces and   $\phi,\psi:(U,\mtop_U) \longrightarrow (V,\mtop_V)$ be independent \loc linear maps. Any two linear combinations of $\phi$ and $\psi$ are also independent. In particular, any linear combination  $\lambda\phi+\mu \psi$ with $\lambda, \mu \in K$ is a \loc linear map.
\mlabel{lem:locallinearmap}
\end{prop}
\begin{proof}
Let $u_1, u_2$ be in $\mtop_U$. Since $\phi$ and $\psi$ are independent \loc linear maps, we have
$  \{\phi(u_1), \psi(u_1)\} \mtop_ V \{\phi(u_2),\psi(u_2)\} $
and hence
$\left(\lambda\phi(u_1)+\mu \psi(u_1)\right)\mtop_V  \left(\lambda\phi(u_2)+\mu \psi(u_2)\right).$
	\end{proof}	

\subsection{\Loc   algebras}

We begin with a preliminary notion. Let $V$ and $W$ be vector spaces and let $\top:=V\times_\top W \subseteq V\times W$. A map $f: V\times_\top W \to U$ to a vector space $U$ is called {\bf \loc bilinear} if
$$f(v_1+v_2,w_1)=f(v_1,w_1)+f(v_2,w_1), \quad f(v_1,w_1+w_2)=f(v_1,w_1)+f(v_1,w_2),$$
$$f(kv_1,w_1)=kf(v_1,w_1), \quad
f(v_1,kw_1)=kf(v_1,w_1)$$
for all $v_1,v_2\in V$, $w_1,w_2\in W$ and $k\in  K $ such that all the pairs arising in the above expressions are in $V\times_\top W$.

\begin{defn}
\begin{enumerate}
\item A {\bf nonunitary \loc  algebra} over $K$ is a \loc vector space $(A,\top)$ over $K$ together with a \loc bilinear map
	$$ m_A: A\times_\top A \to A$$ such that
	$(A,\top, m_A)$ is a \loc semi-group in the sense of Definition~\mref{defn:lsg}.(\mref{it:lsg}).
	\item An {\bf \loc algebra} is a nonunitary \loc algebra $(A,\top, m_A)$ together with a {\bf unit} $1_A:K\to A$ in the sense that
	$(A,\top, m_A, 1_A)$ is a \loc monoid defined in  Definition~\mref{defn:lsg}.~(\mref{defn:partial monoid}). We shall omit explicitly  mentioning the unit $1_A$ and the product $m_A$ unless this generates an ambiguity.
\item
		A  linear subspace $B$ of a \loc algebra $\left(A,\top , m_A \right) $ is called a {\bf sub-\loc algebra} of $A$ if
$B$ is a sub-\loc semigroup of $A$ in the sense of Definition~\mref{defn:lsg}.(\mref{it:lssg}).
\item
A sub-\loc algebra $I$ of a \loc commutative algebra $\left(A,\top ,m_A \right) $ is called a {\bf \loc ideal} of $A$ if for any $b\in I$ we have
$b^\top \cdot b\subseteq I ~\tforall b^\top\in\{b\}^\top$.
\item
An \loc-linear map $f:(A,\mtop_A,\cdot_A)\to (B,\mtop_B,\cdot_B)$ between two (non necessarily unital) \loc algebras is called a {\bf \loc algebra homomorphism} if
\begin{equation}
f(u\cdot_A v)=f(u)\cdot_B f(v)\ \  \tforall (u,v)\in\mtop_A.
\mlabel{eq:lmultlin}
\end{equation}
\label{defn:localisedideal}
\item A \loc algebra $A$ with a linear grading $A=\oplus_{n\geq 0}A_n$ is called a \loc graded algebra if $m_A((A_m\times A_n)\cap \mtop_A) \subseteq A_{m+n}$ for all $m, n\in \Z$.
    \mlabel{it:gradalg}
\end{enumerate}
\mlabel{defn:localisedalgebra}
\end{defn}
It is easy to check that a \loc linear map $f:(A,\mtop_A,\cdot_A)\to (B,\mtop_B,\cdot_B)$ between two \loc algebras is a \loc algebra homomorphism if and only if $\ker f$ is a \loc ideal of $A$,
by the same argument as the one for the corresponding result on an algebra homomorphism.

\begin{rk}
		\begin{enumerate}
\item For a \loc   algebra $(A,\top)$ we have $\{0, 1_A\}\subset A^\top$  since $0\in  A^\top$ by Remark~\mref{rk:locvecsp}.
\item	If $A\times_\top A$ is $A\times A$ in a \loc monoid and \loc algebra, we recover the usual notions of monoid and algebra.
		\end{enumerate}
\mlabel{rk:algebraunit}
\end{rk}

The \loc space $\Q \bfc$, with the multiplication obtained from the linear extension of the \loc monoid structure on $\bfcf$ by the Minkowski product in Eq.~(\mref{eq:minkprod}), is a \loc commutative algebra. By Eq.~\eqref{eq:conegrad}, we in fact have

\begin{lem}
With the grading induced from $\bfcf=\sqcup_{n\geq 0} \bfcf_n$ in Eq.~\eqref{eq:conegrad}, $\Q \bfc=\oplus_{n\geq 0} \Q \bfc_n$ is a graded \loc  algebra.
\mlabel{lem:conegradalg}
\end{lem}

Another important \loc algebra for our purpose is the space $\left({\mathcal M}_\Q, \perp^Q\right)$ with $\perp^Q$ as in Definition \mref{de:meroindep} and the pointwise multiplication. Further by Proposition~\mref{prop:cones}, the linear maps
\begin{equation}
S, I: \Q \bfc\to \calm_\Q
\mlabel{eq:islin}
\end{equation}
linearly extended from those in Eq.~\eqref{eq:conemeromap}, are \loc algebra homomorphisms.

We can say more about the \loc algebra $\calm_\Q$.
By~\cite[Corollary~4.7]{GPZ4}, there is a direct sum decomposition
\begin{equation} \calm _\Q = \calm _{\Q,+}\oplus \calm _{\Q,-}^Q,
\mlabel{eq:merodecomp}
\end{equation}
where $\calm _{\Q,+}$ is the subspace of holomorphic functions and $\calm _{\Q,-}^Q$ is the subspace spanned by polar germs defined by Eq.~\eqref{eq:polar}.
Then we have the following result from  \cite[Corollary 4.18]{GPZ4} reformulated in the terminology of \loc structures.

\begin{prop}
In the decomposition in Eq.~\eqref{eq:merodecomp}, the space  $\calm _{\Q,+}$ is a subalgebra and a \loc subalgebra of $\calm_\Q$. The space $\calm _{\Q,-}^Q$ is not a subalgebra but a \loc subalgebra, in fact a \loc   ideal of $\calm_\Q$. Consequently, the projection $\pi_+^Q:\calm_\Q \to \calm _{\Q,+}$ is a \loc algebra homomorphism.
\mlabel{pp:merodecomp}
\end{prop}

In contrast to the multivariate case, the space  $	{\mathcal M}_{\Q,-}^Q(\R ^*\ot \C )=\e^{-1}\C[\e^{-1}]$ is a subalgebra in the space 	 $	 {\mathcal M}_\Q(\R^*\ot \C  )$  of meromorphic functions in one variable.
	 This is a major difference between our multivariate setup and the usual  single variate framework used for renormalisation purposes. We circumvent the difficulty in  relaxing ordinary  multiplicativity  to a  multiplicativity allowed only on independent elements. In fact, $\calm _{\Q,-}^Q(\R^*\ot \C ) $ is a \loc ideal of $\calm _\Q(\R^*\ot \C  )$ under the restriction of independence relation since the \loc relation $\perp^Q$ restricted to $\calm_\Q(\R^*\ot \C )$ is simply $(\C\times \calm _\Q(\R^*\ot \C )) \cup (\calm _\Q (\R^*\ot \C ) \times \C)$. Thus, the   \loc algebra homomorphism  $\pi_+^Q$ restricts  to  a mere linear map on $\calm _\Q(\R^*\ot \C  )$ with no additional multiplicativity property.

\subsection{\Loc Rota-Baxter algebras and projection maps} The reader is referred to~\cite{G2} for background on Rota-Baxter algebras.
\begin{defn}
A linear operator $P: A\to  A$ on a \loc algebra $(A,\top)$ over a field $ K $ is   called
		 a {\bf \loc Rota-Baxter operator of
		 weight} $\lambda\in  K $, or simply a {\bf Rota-Baxter operator}, if it is a \loc map, independent of $\Id_A$, and satisfies the following {\bf \loc Rota-Baxter relation}:
		 \begin{equation}
 P({ a})\, P({ b})= P(P({ a})\,{ b})+ P({ a}\, P({ b})) +\lambda\, P({ a}\,{ b}) \quad \tforall (a,b)\in \top.
 \mlabel{eq:rbo}
 \end{equation}
We call the triple $(A,\top, P)$ a {\bf \loc Rota-Baxter algebra}.
\mlabel{defn:lrba}
\end{defn}

  \begin{rk}\begin{enumerate}
\item The right hand side of  {Eq.~}\eqref{eq:rbo} is well defined due to the condition that $P$ is independent of the identity.
\item As in the classical setup  \cite[Proposition 1.1.12]{G2}, if  $P$ is a \loc Rota-Baxter operator of weight  $\lambda$, then $-\lambda-P$ is also a \loc Rota-Baxter of weight $\lambda$.
\end{enumerate}
\mlabel{rk:LRB}
\end{rk}

An important class of \loc Rota-Baxter algebras arises from idempotent operators.

\begin{prop}	
Let $\left(A,\top ,m_A \right)$  be  a \loc algebra. Let $P:A\longrightarrow A$ be a \loc linear idempotent  operator in which case there is a linear decomposition $A=A_1\oplus A_2$ with $A_1=\ker\, (\Id-P)$ and $A_2=\ker\,(P)$ so that $P$ is the projection   onto $A_1$ along $A_2$.  The following statements are equivalent:
\begin{enumerate}
\item $P$ or $\Id-P$ is a \loc Rota-Baxter operator of weight $-1$; \mlabel{it:irbo1}
\item $A_1$ and $A_2$ are \loc subalgebras of $A$, and $P$ and $\Id-P$ are independent \loc maps.
\mlabel{it:irbo2}
		\end{enumerate}
If one of the conditions holds, then $P$ is a \loc multiplicative map if and only if $A_2$ is a \loc ideal of $A$.
\mlabel{prop:multpi}
\end{prop}

\begin{proof}
We write $\pi_1=P$ and $\pi_2=\Id -P$.
\smallskip

\noindent
((\mref{it:irbo1}) $\Longrightarrow$ (\mref{it:irbo2}))
It follows from the \loc Rota-Baxter identity (\mref{eq:rbo}) that $A_1=P(A)$ is
a sub-\loc algebra of $A$. Since $\Id-P$ is again an idempotent \loc Rota-Baxter operator, $A_2=(\Id-P)(A)$ is also a sub-\loc algebra of $A$. Finally, $P$  {and} $\Id-P$ are independent \loc maps as a consequence of Definition \mref{defn:lrba}.

\smallskip

\noindent
((\mref{it:irbo2}) $\Longrightarrow$ (\mref{it:irbo1}))
Since $\pi_1$ and $\pi_2=\Id-\pi_1$ are \loc Rota-Baxter operators of weight $-1$ at the same time in view of Remark~\mref{rk:LRB}, we only need to verify that $\pi_1$ is a \loc Rota-Baxter operator of weight $-1$:
\begin{equation}\mlabel{eq:rbeminus}
\pi_1({ a})\,\pi_1({ b})+\pi_1(a\,b)= \pi_1(\pi_1({ a})\,{ b})+ \pi_1({ a}\, \pi_1({ b}))\quad \tforall ({ a}, { b})\in \top.
\end{equation}

Write $a=a_1+a_2$ and $b=b_1+ b_2$. Since the projections $\pi_i, i=1,2,$ are independent \loc maps, it follows that $\{a_1,a_2\} \top \{b_1,b_2\}$. Thus every term in
$$ ab = a_1b_1 + a_1b_2+a_2b_1 +a_2b_2$$
is well defined, with $a_1b_1\in A_1$ and $a_2b_2\in A_2$.
Then the left hand side of Eq.~(\mref{eq:rbeminus}) becomes
$$ a_1b_1+\pi_1(a_1b_1 +a_1b_2+a_2b_1 +a_2b_2)
=2a_1b_1+\pi_1(a_1b_2)+\pi_1(a_2b_1).$$
The right hand side of Eq.~(\mref{eq:rbeminus}) becomes
$$\pi_1(a_1b)+\pi_1(ab_1)=\pi_1(a_1b_1+a_1b_2)+\pi_1(a_1b_1+a_2b_1)
=\pi_1(a_1b_2)+\pi_1(a_2b_1)+2a_1b_1,$$
as needed.
Then the last statement follows from the remark after Definition~\mref{defn:localisedalgebra}.
\end{proof}

As an immediate consequence of Proposition~\mref{pp:merodecomp} and \mref{prop:multpi}, we obtain

\begin{coro} The  projection $\pi_+^Q$ onto $\calm _{\Q,+}$ along $\calm^Q_{\Q -}$ is \loc multiplicative and  $(\calm _{\Q},\pi_-^Q)$ is a \loc Rota-Baxter algebra.
\mlabel{co:merorbo}
\end{coro}

\section{Locality for coalgebras and the  {Locality Conservation} Principle}
\mlabel{sec:lcoalg}

 We introduce the concept of a \loc coalgebra  which involves a suitable \loc tensor product. Between a \loc coalgebra and a \loc algebra, we consider \loc convolution and \loc convolution inverse when the \loc coalgebra is connected. We then prove our  Locality  {Conservation} Principle, showing that a locality (independent) relation is preserved by the renormalisation procedure \`a la Connes and Kreimer i.e.,  carried out by means of   algebraic Birkhoff factorisation in the coalgebra-algebraic context~\mcite{GPZ2}.

\subsection{\Loc tensor product}
We first give a \loc version of the tensor product by considering a relative generalisation of the \loc relation. As noted before, vector spaces and tensor products are taken over the base field $ K $ unless otherwise specified.

Let $V$ and $W$ be vector spaces and let $\top \subseteq V\times W$. For $X\subseteq V$ and $Y\subseteq W$, denote
$$ X^\top: =\{w\in W\,|\, X\top w\}, ^\top Y:=\{v\in V\,|\, v\top Y\}.$$
A {\bf relative \loc vector space} is a triple $(V,W,\top)$ where $V$ and $W$ are vector spaces and $\top:=V\times_\top W$ is a subset of $V\times W$ such that for any sets $X\subseteq V$ and $Y\subseteq W$, the sets $X^\top$ and $^\top Y$ are linear subspaces of $V$ and $W$ respectively.

\begin{ex} Given vector spaces $V$ and $W$, any subspaces $V_1\subset V$ and $W_1\subset W$ give rise to a relative \loc vector space $(V, W, \top)$ with $\top= (V_1\times W_1)\cup (\{0\}\times W)\cup (V\times \{0\})$.
\end{ex}
  Given two vector spaces (resp. a relative \loc vector space $(V,W,\top)$), define $I_{\mathrm{bilin}}$ (resp. $I_{\top,\mathrm{bilin}}$) to be the subspace of $  K   (V\times W)$ (resp. $  K   (V\times_\top W)$) spanned by the bilinear relations
\begin{equation}
\begin{array}{c}
(v_1+v_2,w_1)-(v_1,w_1)-(v_2,w_1), \ (v_1,w_1+w_2)-(v_1,w_1)-(v_1,w_2), \\
 (kv_1,w_1)-(v_1,kw_1), \ (kv_1,w_1)-k(v_1,w_1)
\end{array}
\mlabel{eq:bilinear}
\end{equation}
for all $v_1,v_2\in V, w_1,w_2\in W, k\in   K  $ (resp. such that the pairs in Eq.~(\mref{eq:bilinear}) are in $\top$).

\begin{coex} Take $\{e_1, e_2, e_3, e_4\}$ to be the canonical orthonormal basis of $V:=\R^4$    and \begin{eqnarray*}V\times_\top V&= &\left(\left(V\times\{0\}\right)\cup \left(\{0\}\times V\right)\right)\,\cup \left(K \,(e_1+e_2)\times K\,e_3\right)\\
&  \cup &\left(K \,(e_1+2e_2)\times K\,e_4\right)\,\cup\,\left(K \,e_1 \times K\,(e_3+e_4)\right)\,\cup\,\left(K \, e_2\times K\,(e_3+2 e_4)\right).
\end{eqnarray*}
Then   \[ (-e_1-e_2, e_3)+ (-e_1-2e_2, e_4)+ (e_1, e_3+e_4)+ (e_2, e_3+2 e_4)\]   is an element of $K(V\times_\top W)\cap I_{\mathrm{bilin}}$ as can easily be seen   using the   defining relations for $I_{\mathrm{bilin}}$, but it is not in $I_{\top,\mathrm{bilin}}$.

\end{coex}

So in general
$$ K   (V\times_\top W) \cap I_{\mathrm{bilin}}\supsetneqq I_{\top,\mathrm{bilin}},
$$
and thus
$$ h: K   (V\times_\top W)/I_{\top,\mathrm{bilin}}\to V\otimes W
$$
is not injective. Therefore
$$ K (V\times_\top W)/I_{\top,\mathrm{bilin}}
$$
is not an appropriate candidate for a \loc tensor product, since the image of a \loc coproduct $\Delta$ should  {lie}  in $V\ot V$.
Instead, we set
\begin{equation}
V\otimes_\top W : = \mathrm{im} \,h\subset V\ot W,
\mlabel{eq:bt}
\end{equation}
which is also the image of the composition map
$ K(V\times_\top W )\to   K( V\times W)\to V\otimes W.$

\subsection{\Loc coalgebras}
We recall that
a coalgebra $(C,\Delta)$ over a field $   K   $ is  {\bf counital} if there is a map $\vep:C\to  K  $ such that $(\vep\otimes \Id _C)\Delta =(\Id _C\otimes \vep)\Delta=\Id _C$. It is {\bf  ($\Z_{\geq 0}$-)graded}   if
$$C=\bigoplus_{n\in\Z_{\geq 0}}C_n \quad \text{and} \quad \Delta(C_n)\subseteq\bigoplus_{p+q=n} C_p\otimes C_q, \quad \bigoplus_{n\geq 1} C_n\subseteq \ker \vep.$$
Thus $C=C_0+\ker \vep$.
Moreover a graded coalgebra is
called {\bf connected}  if  $C=C_0\oplus \ker \vep$. Consequently, $\vep$ restricts to a linear bijection $\vep: C_0\cong  K $ and $\ker \vep = \oplus_{n\geq 1}C_n$.

For the sake of simplicity, we shall drop the explicit mention of the grading and simply call such a coalgebra  a connected coalgebra.

The following definition is dual to that of a \loc   algebra. As in Eq.~(\mref{eq:bt}), let
$$C^{\otimes_\top n}=\underbrace{C\otimes_\top \cdots \otimes_\top C}_{n \text{ factors }}$$
denote the image of $  K   C^{_\top n}$ in $C^{\ot n}$.

\begin{defn}
Let $(C,\top)$ be a \loc vector space and let $\Delta:C\to C\otimes C$ be a linear map. $(C,\top, \Delta)$ is a {\bf \loc noncounital coalgebra} if  it satisfies the following two conditions
\begin{enumerate}
\item for any $U\subset C$ (compare with Eq.~\eqref{eq:semigrouploc})
\begin{equation}
\Delta (U^\top) \subset U^\top \otimes _\top U^\top.
\mlabel{eq:comag}
\end{equation}
In particular, $\Delta(C)\subseteq C\ot_\top C$;
\mlabel{colocalDelta} 	
\item     the following {\bf coassociativity}  holds:
$$(\Id _C\ot \Delta)\,\Delta= (\Delta\ot \Id _C)\, \Delta.$$
\mlabel{localDelta}
\end{enumerate}
\begin{itemize}
\item
If in addition, there is a {\bf counit}, namely a linear map $\vep: C\to   K  $ such that
  $ ({\rm Id}_C\otimes   \vep)\, \Delta=  (\vep\otimes {\rm Id}_C)\, \Delta= {\rm Id}_C$, then $(C,\top,\Delta,\vep)$ is called a {\bf \loc coalgebra.}
\item A {\bf connected \loc coalgebra} is a \loc coalgebra $(C,\top, \Delta)$ with a grading $C=\oplus_{n\geq 0} C_n$ such that, for any $U\subseteq C$,
\begin{equation}
\Delta(C_n\cap U^\top)\subseteq\bigoplus_{p+q=n} (C_p\cap U^\top)\otimes_\top (C_q\cap U^\top), \quad \bigoplus_{n\geq 1} C_n = \ker \vep.
\mlabel{eq:lconn}
\end{equation}

We denote  by $J$   the unique element of $C_0$ with $\vep (J)=1_ K $, giving $C_0= K \, J$.
\end{itemize}
\mlabel{defn:colocalcoproduct}
\end{defn}

\begin{rk}
We note that whereas the conditions for a \loc algebra are weaker than those for an algebra, the conditions for a \loc coalgebra are stronger than those for a coalgebra. In particular, a \loc coalgebra is a coalgebra and a connected \loc coalgebra is a connected coalgebra.
\mlabel{rk:conn}
\end{rk}

A natural example of \loc coalgebra is given by the coalgebra $\Q \bfcf$ of lattice cones in~\mcite {GPZ2}. Let $(C, \Lambda _C)$ be a lattice cone, $(F, \Lambda _F)$ a face of $(C, \Lambda _C)$, which we denote by $(F, \Lambda _F)<(C, \Lambda _C)$. The transverse cone $t((C, \Lambda _C),(F, \Lambda _F))$ is the orthogonal projection of $(C, \Lambda _C)$ to the orthogonal subspace of the subspace spanned by $F$. Then the coproduct $\Delta(C, \Lambda _C)$ of $(C, \Lambda _C)$ is defined by
\begin{equation}
\Delta ((C, \Lambda _C))=\sum _{F<C} t((C, \Lambda _C),(F, \Lambda _F))\otimes (F, \Lambda _F).
\mlabel{eq:conecoprod}
\end{equation}
Since $t((C, \Lambda _C),(F, \Lambda _F)) \perp ^Q (F, \Lambda _F)$ by definition, the quadruple $(\Q \bfc, \perp ^Q, \Delta, \vep)$ is a \loc coalgebra with the \loc counit given by the linear extension of
the map
$$ \vep: \bfcf \to \Q , (C, \Lambda_C)\mapsto \left\{\begin{array}{ll} 1, & (C, \Lambda_C)=(\{0\},\{0\}), \\ 0, & (C, \Lambda_C)\neq (\{0\},\{0\}). \end{array} \right . $$
Further the connectedness conditions in Eq.~\eqref{eq:lconn} are satisfied. Thus we have proved
\begin{lem}
$(\Q \bfc, \perp^Q, \Delta, \vep)$ with the grading $\Q \bfc = \oplus_{n\geq 0}\Q \bfc_n$ from Eq.~\eqref{eq:conegrad} is a connected \loc coalgebra.
\mlabel{lem:conecograd}
\end{lem}

Let us list a few useful general properties of \loc coalgebras.
\begin{lem}
Let $(C,\top,\Delta)$ be a \loc coalgebra.
\begin{enumerate}
\item For any $n\geq 2$ and $0\leq i \leq n$,
\begin{equation} \mlabel{eq:gcoasso}
\id _C^{\ot i} \ot \Delta \ot \id _C^{\ot (n-i-1)}: C^{\ot _\top n} \to C^{\ot_\top (n+1)}
\end{equation}
\mlabel{it:coalgId}
\item
We have
$(\Delta\ot \Delta)(C\otimes_\top C)\subseteq C^{\ot_\top 4}$;
\mlabel{it:coalgmap1}
\item
We have
$(\Delta\times \Delta)(C\times_\top C)\subseteq (C\otimes_\top C)\times _\top (C\otimes_\top C)$;
\mlabel{it:coalgmap3}
\item
For any \loc linear map $\phi$ independent of $\id$, $n\geq 2$ and $0\leq i \leq n$, we have
$\id^{\ot i}\ot \phi\ot \id^{\ot (n-i-1)}: C^{\ot_\top n} \to C^{\ot_\top n}$. \mlabel{it:coalgmap2}
\end{enumerate}
 \mlabel{lem:coalgmap}
 \end{lem}

\begin{proof}
(\mref{it:coalgId}) Let $n\geq 2$ and $1\leq i\leq n$ be given. By the definition of $C^{\ot_\top n}$, any of its elements is a sum of pure tensors $c_1\ot \cdots \ot c_n$ with $(c_1,\cdots,c_n)\in C^{_\top n}$. Let $U=\{c_j\,|\, j\not =i+1\}$. Then $c_{i+1}\in U^\top$,
so there exist $(d_1, e_1), \cdots, (d_k,e_k)\in U^\top\times _\top U^\top$, such that
$$\Delta (c_{i+1})=\sum_{ \ell} d_\ell \otimes e_\ell.
$$
Now
$$(\id _C^{\ot i} \ot \Delta \ot \id _C^{\ot (n-i-1)})(c_1\ot \cdots \ot c_n)=\sum_{ \ell} c_1\ot \cdots c_i \otimes d_\ell\ot e_\ell\ot c_{i+2}\ot \cdots \ot c_n,
$$
and $c_1\ot \cdots c_i \otimes d_\ell\ot e_\ell\ot c_{i+2}\ot \cdots \ot c_n\in C^{\ot_\top (n+1)}$.

\noindent (\mref{it:coalgmap1}) Since $(\Delta\ot \Delta)=(\Delta\ot \id)(\id \ot \Delta)$, from Eq.~(\mref{eq:gcoasso}) we obtain
\begin{eqnarray*}
(\Delta\ot \Delta)(C\otimes_\top C)&=& (\Delta\ot \id)(\id \ot \Delta) (C\otimes_\top C)\\
&\subseteq & (\Delta\ot \id)(C\otimes_\top C \otimes_\top C)\\
&\subseteq & C\otimes_\top C\otimes_\top C\otimes_\top C.
\end{eqnarray*}

\noindent (\mref{it:coalgmap3})
Let $(c_1,c_2)\in C\times_\top C$. Then $c_2\in \{c_1\}^\top$. So by Eq.~(\mref{eq:comag}), $\Delta(c_2)=\sum_{(c_2)} c_{2,(1)}\ot c_{2,(2)}$ with $c_{2,(1)}\top c_{2,(1)}$ and $\{c_{2,(1)}, c_{2,(2)}\}\subseteq \{c_1\}^\top.$ Thus $c_1 \in \{c_{2,(1)},c_{2,(2)}\}^\top$. By Eq.~(\mref{eq:comag}) again, $\Delta(c_1)=\sum_{(c_1)} c_{1,(1)}\ot c_{2,(2)}$ with $c_{1,(1)}\top c_{1,(2)}$ and $\{c_{1,(1)},c_{1,(2)}\} \subseteq \{c_{2,(1)},c_{2,(2)}\}^\top.$ This shows that $(c _{1,(1)},c _{1,(2)}, c _{2,(1)}, c _{2,(2)})$ is in $C^{_\top 4}$ and hence $ \big((c _{1,(1)}\ot c _{1,(2)}), (c _{2,(1)}\ot c _{2,(2)})\big)$ is in $(C\otimes_\top C)\times_\top (C\otimes_\top C)$.
\smallskip

\noindent
(\mref{it:coalgmap2}) Again any element of $C^{\ot_\top n}$ is a sum of pure tensors $c_1\ot \cdots \ot c_n$ with $(c_1,\cdots,c_n)\in C^{_\top n}$. Thus $(c_1,\cdots,c_{i-1}, \phi(c_i),c_{i+1},\cdots,c_n)$ is in $C^{_\top n}$. This is what we want since $(\id^{\ot i}\ot \phi\ot \id^{\ot (n-i-1)})(c_1\ot\cdots\ot c_n)=
c_1\ot\cdots\ot c_{i-1}\ot \phi(c_i)\ot c_{i+1}\ot\cdots\ot c_n$.
\end{proof}

\begin{lem} \label{lem:reduced_coproduct}
Let $(C=\oplus_{n\geq 0} C_n,\top,\Delta)$ be a connected \loc coalgebra. Define the {\bf reduced coproduct}
$\tilde{\Delta}(c):= \Delta(c)- J\ot c-c\ot J$.
Recursively define
\begin{equation}
\tilde{\Delta}^{(1)}=\tilde{\Delta}, \quad \tilde{\Delta}^{(k)}:=(\id\ot \tilde{\Delta}^{(k-1)})\tilde{\Delta}, \ \  k\geq 2.
\mlabel{eq:redcoprod}
\end{equation}
\begin{enumerate}
\item
For $c\in \oplus_{n\geq 1}C_n$,
$\tilde \Delta (c)= \sum_{(c)} c'\ot c''$ with $\deg(c'), \deg(c'')>0$ and $(c',c'')\in C\times_\top C$;
\mlabel{it:conil-1}
\item
If in addition  $c\in U^\top $ for some $U\subset C$, then the above pairs
$(c', c'')$ are in $U^\top\times _\top U ^\top$;
\mlabel{it:conil0}
\item
$\tilde\Delta^{(k)}(x)$ is in $C  ^{\ot _\top (k+1)}$ for all $x\in  C, k\in\N$;
\mlabel{it:conil1}
\item
$\tilde \Delta^{(k)}(   C  _n)=\{0\}$ for all $k\geq n.$
\mlabel{it:conil2}
\end{enumerate}
\mlabel{lem:conil}
\end{lem}

\begin{proof}
We only need to prove (\mref{it:conil0}) since then (\mref{it:conil-1}) is the special case when $U=\{0\}$.

Let $c\in C_n\cap U^\top$. By {Eq. \eqref{eq:lconn}},
we can write
$$ \Delta(c)=y\ot J + J\ot z + \sum_{(c)}c'\ot c''$$
with $y, z\in C_n$, $c', c''\in U^\top$ and each $c'\ot c''\in C_p \otimes _ \top C_q, p+q=n, p, q\geq 1$. Then by the same argument for a connected coalgebra~\cite[Theorem~2.3.3]{G2}, we obtain $y=z=x$. This proves (\mref{it:conil0}).

Then (\mref{it:conil1}) follows from an easy induction on $k$ by the \loc property of $\Delta$; while the proof of (\mref{it:conil2}) is similar to the case without a \loc structure~\cite[Proposition II.2.1]{M}.
\end{proof}

\subsection{Locality of the convolution product}

We show that the locality (independence) of linear maps are preserved under the convolution product.

\begin{lem}  Let $(C,\mtop_C )$ be a \loc coalgebra with counit $\vep_C: C\to  K   $. Let $(A,\mtop_A )$  be a \loc   algebra with unit $u_A:    K   \to A$.
The map  $e:=   u_A   \vep_{C}:C\to A$ is independent to any linear map $\phi: C\longrightarrow A$. In particular, the map  $e$ is a \loc linear map.
\mlabel{lem:mutuallyindIA}
\end{lem}
\begin{proof} This is because   $\mathrm{im}\, e=  K  \cdot 1 _A \subset A ^{\mtop_A}$ as we can see from Remark \mref{rk:algebraunit}.
\end{proof}

Now we give a general result.
\begin{prop}
 Let $\left( C ,\mtop_C , \Delta \right)$   be a \loc coalgebra and let $\left(A ,\mtop_A,m_A\right) $ be a \loc algebra. Let ${\mathcal L}:=\text{\tloc}\Hom(C,A)$  be the space of \loc linear maps.
  Define
    $$ \mtop_{\mathcal L}:=\left\{ (\phi,\psi)\in \call\times \call\,|\, (\phi\times \psi)(C\times_\top C)\subseteq A\times_\top A\right\}.$$
\begin{enumerate}
\item
For $(\phi,\psi)\in \mtop_\call$, define the {\bf convolution product} of $\phi$ and $\psi$ by
\begin{equation}
\phi\star \psi: C \stackrel{\Delta_C}{\longrightarrow} C\otimes_\top C \stackrel{\phi\ot \psi}{\longrightarrow} A\otimes_\top A \stackrel{m_A}{\longrightarrow} A.
\mlabel{eq:lconv}
\end{equation}
Then $\phi\star \psi$ is a \loc linear map and
the triple $({\mathcal L}, \mtop_{\mathcal L}, \star)$ is a \loc algebra.
\mlabel{it:conv1}
\item
If moreover    $C$ is connected  then
$$\calg \call:=\{\phi\in \call\ | \ \phi (J)=1_A\}$$ is  a \loc  group for the convolution product.
\mlabel{it:conv2}
\item
Under this same assumption, we have \[(\phi, \psi)\in \mtop_{\mathcal L}\cap\left(\calg \call\times \calg \call\right) \Longrightarrow (\phi^{\star k}, \psi^{\star l})\in \mtop_{\mathcal L}\cap \left(\calg \call\times \calg \call\right) \quad \tforall k, l\in \Z.\]	
\mlabel{it:conv3}
\end{enumerate}
\mlabel{prop:inversemap}
 	\end{prop}

\begin{proof}
(\mref {it:conv1}) For $(\phi, \psi)\in \mtop_{\mathcal L}$, since $(\phi\times \psi)(C\times_\top C)\subseteq A\times_\top A$, we have
 $$(\phi\ot \psi)(C\otimes_\top C)\subseteq A\otimes_\top A.
 $$
Hence the composition in Eq.~(\mref{eq:lconv}) is well-defined, giving a well-defined convolution product.

 We next verify that $\phi\star\psi$ is a \loc linear map. For $ c  _1\mtop_C\,  c _2$,
 by Lemma~\mref{lem:coalgmap}.(\mref{it:coalgmap3}), there are finitely many $(d_i,e_i), (f_j,g_j)\in C\times_\top C$ with $(d_i, e_i, f_j, g_j)\in C^{_\top 4}$, such that
$$\Delta ( c  _1)=\sum_{i} d  _{i}\ot  e  _{i} \quad \text{ and} \quad \Delta ( c  _2)=\sum_{j}   f_{j}\ot  g  _{j}.
 $$
 Then
$$
 \phi\star\psi( c _1)=\sum_{i} \phi( d_i)\psi( e_i) \quad
 \text{and }\quad
 \phi\star\psi( c _2)=\sum_{j} \phi( f_j)\psi( g_j).
$$
From $(\phi, \psi)\in \mtop_{\mathcal L}$, we obtain
$\left(\phi( d_i),\psi( e_i),\phi( f_j),\psi( g_j)\right)\in A^{_\top 4}$. So
$$\big(\sum_i \phi( d_i)\psi( e_i),\sum_j \phi( f_j)\psi( g_j)\big)$$
is in $A\times_\top A$ and
thus $(\phi\star\psi( c _1))\mtop_A(\phi\star\psi( c _2))$.

Thus we are left to verify the axioms for a \loc semigroup: the closeness of $U^{\mtop_{\mathcal L}}$ under the convolution product for every $U\subseteq {\mathcal L}$ and the associativity.

Let $\psi$ and $\chi$ be independent \loc linear maps in $U^{\mtop_\call}$ and let $\phi$ be in $U$. Then
$\phi,\psi,\chi$ are pairwise independent. Therefore
$$\phi\times \psi\times \chi: C^{_\top 3} \longrightarrow A^{_\top 3}$$
is well defined.
For $(c_1,c_2)\in C^{_\top 2}$, that is $c_1\in \{c_2\}^\top$, there exist
$(d_1, e_1), \cdots, (d_k,e_k)\in \{c_2\}^\top \times _\top \{c_2\}^\top $, such that
$$\Delta ( c  _1)=\sum_{i} d  _{i}\ot  e  _{i},
 $$
 with
$(d_i, e_i, c_2)\in C^{_\top 3}$.
Then
$(\psi( d_i), \chi(e_i), \phi (c_2))\in A^{_\top 3}$ and hence $(\psi( d_i)\chi( e_i))\mtop_A \phi (c_2).$
So we have
$(\psi\star\chi)( c _1)\mtop_A \phi (c_2),
$
which means
$ \psi\star \chi$ is in $\phi^{\mtop_{\mathcal L}}.$
Thus $\psi \star \chi \in U^{\mtop_\call}$. This verifies the first axiom. The associativity of $\star$ follows from the associativity of $m$ and coassociativity of $\Delta$ as in the classical case.
\smallskip

\noindent
(\mref{it:conv2}) We now assume that $ C $  is a connected   \loc coalgebra.
For a \loc linear map $\phi : C \to A$, we now prove by induction on the degree of $ c  _1$ that the map
\begin{equation}
\phi ^{\star(-1)}( c  _1)=\left\{\begin{array}{ll}
1_A, & c_1=J, \\
-\phi ( c _1)-\sum_{( c  _1)} \phi ( c  _1^\prime) \phi ^{\star(-1)}(  c  _1^{\prime \prime}), & c_1\in \ker \vep, \end{array} \right .
\mlabel{eq:phirec}
\end{equation}
is well defined, and that $c _1 \mtop_C\,  c$ implies $\phi ^{\star(-1)}( c  _1)\mtop_A \phi ( c ).$

{This} is trivial for degree $0$ since $\phi^{\star (-1)}(J)=1_A$. Assume for any $ c  _1\in C$ of degree $\le n$, $\phi ^{\star(-1)}( c  _1)$ is well defined, and for $c$ with $ c  _1 \mtop_C\,  c $,
$\phi ^{\star(-1)}( c _1)\mtop_A \phi ( c )$ holds.

Now for any $c_1$ of degree $n+1\ge 1$ with $ c  _1 \mtop_C\,  c  $, by Lemma~\mref {lem:conil}.(\mref{it:conil0}), we have
 $$\Delta ( c  _1)= c  _1\ot J+J\ot  c  _1+\sum _{( c  _1)} c  _1^\prime \ot  c  _1^{\prime \prime}
\ \text{ with }
 (c  _1',  c  _1^{\prime \prime}, c)\in C^{_\top 3}.
 $$
By the induction hypothesis, $\phi ^{\star(-1)}( c _1^{\prime\prime})$ is well defined, such that $\phi ^{\star(-1)}( c _1^{\prime\prime})\mtop_A \phi (c)$ and $\phi ^{\star(-1)}( c _1^{\prime\prime})\mtop_A \phi (c_1^\prime)$. Since $\phi$ is a \loc linear map, we also have
$\phi ( c  _1)\mtop_A \phi ( c  )$ and $ \phi ( c  _1^\prime )\mtop_A \phi ( c  ).$
Thus
$\phi ( c  _1^\prime) \phi ^{\star(-1)}(  c  _1^{\prime \prime})$ is well defined and
$(\phi ( c  _1^\prime) \phi ^{\star(-1)}(  c  _1^{\prime \prime}))\mtop_A \phi ( c ).
$
So, $\phi ^{\star(-1)}(  c  _1)$ is well defined and
$\phi ^{\star(-1)}( c  _1)\mtop _ A \phi ( c )$, which means $\phi \mtop_\call \phi ^{\star(-1)}$.

Again by induction on the degree of $c_1$, we now  prove that
$\phi ^{\star(-1)}$ is a \loc linear map by checking
\begin{equation}
\phi ^{\star(-1)}( c  _1)\mtop_A \phi ^{\star(-1)}( c  _2)\quad \tforall c_2\in C, c_1\mtop_C c_2,
\mlabel{eq:philoc}
\end{equation}
a fact which is obvious at degree 0. Assume that, for a given $n\geq 0$ and any $c_1$ of degree $\le n$
the equation holds. Consider $ c  _1 $ of degree $n+1\ge 1$. Since $c_1\mtop_C c_2$, we can choose
 $$\Delta ( c  _1)= c  _1\ot J+J\ot  c  _1+\sum _{( c  _1)} c  _1^\prime \ot  c  _1^{\prime \prime}
 $$
such that
$ \{c_1,c_1^\prime, c_1''\}\mtop_C\, c_2.$
From this we have $\{\phi ( c  _1), \phi ( c  _1'),
\phi ^{\star(-1)}( c  _1'')\} \mtop_A\, \phi ^{\star(-1)} ( c  _2).$
So Eq.~(\mref {eq:phirec}) gives
$\phi ^{\star(-1)}( c  _1)\mtop_A \phi ^{\star(-1)}( c  _2).$
Therefore, we conclude that $\calg \call $ is a \loc group with unit $u_A \varepsilon_C$ by Example \mref{lem:mutuallyindIA}.
\smallskip

\noindent
(\mref{it:conv3}) Similar inductions show that
$\phi\mtop_{\call}\psi$ implies $\phi \mtop_{\call} \psi^{\star (-1)}$ and $\phi^{\star (-1)}\mtop_{\call} \psi^{\star (-1)},$
so that $\phi^{\star k}\mtop_{\call} \psi^{\star l}$ for any $k,l\in \Z$.
 \end{proof}

\subsection{The   Locality  {Conservation} Principle}
\mlabel{ss:lpcoalg}
We now put together all the \loc structures we have obtained so far to provide an answer to the question addressed by the   Locality {Conservation} Principle proposed in Problem~\mref{pr:lpcoalg}. It is formulated as the preservation of locality by the algebraic Birkhoff factorisation in the coalgebra context in~\cite[Theorem~4.4]{GPZ2}.
See Section~\mref{ss:lphopf} for the formulation of the \Loc Product Conservation Principle built on the algebraic Birkhoff factorisation in the Hopf algebra context originated from Connes-Kreimer~\mcite{CK}.

\begin{theorem}
{\bf (Algebraic Birkhoff  factorisation, \loc coalgebra version) }
Let $\left( C ,\mtop_C , \Delta \right)$   be a connected \loc coalgebra, $C=\oplus_{n\geq 0} C _n$, $C _0= K  J$.
Let $\left(A,\mtop_A,\cdot \right) $ be a commutative \loc algebra with decomposition $A=A_1\oplus A_2$ as a vector space satisfying the following
\begin{quote}
{\bf (Basic Assumption)}
the linear projections $\pi_i$ onto $A_i$ along $A_{\hat{i}}, \{\hat{i}\}:=[2]\backslash \{i\}, i=1,2,$  are independent \loc linear maps and $1_A$ is in $A_1$.
\end{quote}
Let
\begin{equation*}
 \phi: \left(C , \mtop_C \right)\longrightarrow   \left(A,\mtop_A\right)
\end{equation*}
be a \loc linear map such that $\phi (J)=1_A$.
Then there are unique independent \loc linear maps $\phi_i: C\to K+A_i$ with $\phi_i(J)=1_A$ and $\phi_i (\ker \vep)\subseteq A_i, i=1,2, $ such that

\begin{equation}
\phi= \phi_1^{\star (-1)} \star \phi_2.
\mlabel{eq:lcabf}
\end{equation}
The map $(\phi_1)^{\star(-1)}$ is a \loc linear map and $\phi_1\mtop_\call \{\phi, \phi_2\}$, $\phi_1^{\star(-1)}\mtop_\call \{\phi _1,  \phi _2\}$.

\begin{enumerate}
\item
If in addition to the Basic Assumption, $A_1$ is a sub-\loc algebra of $A$, then $\phi_1^{\star (-1)}:C\to  K  + A_1$.
\mlabel{it:add-subal}

\item
If in addition to the Basic Assumption and Item~(\mref{it:add-subal}), $A_2$ is a \loc ideal of $A$, then $\phi_1^{\star (-1)}=\pi_1  \phi$
{and $\phi_2$ is recursively given by
 {
\begin{equation}
\phi_2(J)=1_A, \ \phi_2(c)=(\pi_2  \phi)(c)-\sum_{(c)}(\pi_1  \phi)(c')\phi_2(c'') \tforall c\in \ker \vep,
\mlabel{eq:phi2}
\end{equation}
with $\tilde\Delta(c)=\sum_{(c)}c'\ot c''$} the reduced coproduct defined in Lemma \ref{lem:reduced_coproduct}}.
\mlabel{it:add-ideal}
\end{enumerate}
If $\psi: \left(C , \mtop_C \right)\longrightarrow   \left(A,\mtop_A\right)
$
is also a \loc linear map independent of $\phi$ with $\psi(J)=1_A$, then $\phi_i$ and $\psi_j$ are independent for $i, j =1,2$.
\mlabel {thm:abflcoalg}
\end{theorem}

\begin{rk}
Theorem~\mref{thm:abflcoalg} provides an answer to Problem~\mref{pr:lpcoalg}: the locality of the renormalised map follows from that of the initial map under the assumption of the theorem. See also Remark~\mref{rk:motto} and its subsequent example.
\end{rk}

\begin{proof} The proof of the uniqueness of the maps $\phi_i$ is the same as the proof \cite[Theorem 4.4]{GPZ2} for the case of a trivial \loc relation i.e., when $\top= C\times C$.

Let $n\geq 1$ and $ c  \in C_n$. Since $C$ is a connected \loc coalgebra, we can write
$$\Delta ( c  )=J\ot  c  + c  \ot J+\sum _{( c  )} c '\ot c ''
$$
with $\deg(c'), \deg(c'')>0$ and $c'\mtop_C c''$.

We first prove by induction on the degree $n$ of $c$ that the map given by
\begin {equation}
\phi_1( c )=\left\{\begin{array}{ll}
1_A, & c=J, \\
-\pi _1\Big(\phi( c )+\sum_{( c )} \phi_1( c ')\phi( c  '')\Big),
& c\in C_n, n>0, \end{array} \right .
\mlabel{eq:LPhi-}
\end{equation}
is well-defined, and for any $d\in C$ with $d \mtop _C  c $, there is
$$\phi _1( c )\mtop _A \phi (d),
$$
which   {clearly hold} for $ c $ of degree $0$.

Assume that these hold true for $ c $ of degree $\le n$. Then for $ c  $ of degree $n+1$, $ c  '$ is of degree $\le n$, so $\phi _1( c  ')$ is defined and $\phi _1( c  ')\mtop_A \phi ( c  '')$. Therefore $\phi _1( c  ')\phi ( c  '')$ makes sense, and $\phi _1( c )$ is well-defined.

Further for any $d \in C$ with $ c  \mtop_C d$, we can take $\{c',c''\}\mtop_C d$. Since $\phi $ is a \loc map, we obtain
$\phi ( c  )\mtop_A \phi (d)$ and $\phi ( c''  )\mtop_A \phi (d)$ . Also the induction hypothesis gives
$\phi _1( c  ')\mtop_A \phi (d).$
Thus
$(\phi _1( c  ')\phi ( c  ''))\mtop_A \phi (d).$
Therefore,
$$\Big(\phi( c )+\sum_{( c )} \phi_1( c ')\phi( c  '')\Big)\mtop_A \phi (d).
$$
Now since $\pi _1$ is a \loc map and $\pi _1 $ and $\pi _2$ are independent, $\pi _1$ and $\Id _A=\pi _1+\pi _2$ are independent. Thus
$\phi_1( c )\mtop_A \phi (d).
$
Therefore we have proved that $\phi _1$ is well-defined and $\phi _1\mtop_\call \phi$.

Now for any $ c  \mtop_C d$, we have $\phi (c) \mtop_A \phi_1(d)$. By a similar induction on the degree of $c$, we obtain
$ \phi_1( c )\mtop_A \phi _1(d),
$
so
$\phi _1$ is a \loc linear map.
Therefore, the map
\begin{equation}
\phi_2(c):=\left\{\begin{array}{ll}
1_A, & c= J, \\
(\Id_A-\pi _1)\big(\phi(c)+\sum_{(c)} \phi_1(c')\phi(c'')\big), &
c\in C_n, n>0, \end{array} \right .
\mlabel{eq:LPhi+}
\end{equation}
is well-defined.

Notice that for $c\in C_n, n>0$,  Eq.~(\mref{eq:LPhi+}) means
\begin {equation}
\phi_2(c)=\phi (c)+ \phi _1(c)+\sum_{(c)} \phi_1(c')\phi(c'').
\mlabel{eq:LPhi2}
\end{equation}
With the condition on $J$, this in turn reads
$\phi _2=\phi _1\star \phi$ and hence $\phi =  \phi_1^{\star(-1)}\star \phi _2.$
By Proposition~\mref {prop:inversemap},  $  \phi_1^{\star(-1)}$ is a \loc linear map. From Eq.~(\mref {eq:LPhi2}), we easily obtain
$\phi _1\mtop_\call \phi _2.
$
By Eq.~(\mref {eq:phirec}), an easy induction on the degree of $c$ shows
$\phi_1^{\star(-1)}\mtop_\call\, \{\phi _1, \ \phi _2\}.$

A similar induction shows that if $\psi: \left(C , \mtop_C \right)\longrightarrow   \left(A,\mtop_A\right)$
is also a \loc map, independent of $\phi$ with $\phi(J)=1_A$, then $\phi_i$ and $\psi_j$ are independent for $i, j =1,2$, proving the last statement of the theorem.

Now to prove (\mref{it:add-subal}), letting $A_1$ be a sub-\loc algebra,  Eq.~(\mref {eq:phirec}) and a simple induction on $n\geq 0$ show
that $\phi_1^{\star (-1)}(c) \in  K  + A_1$ for any $c\in C_n$.

To prove (\mref{it:add-ideal}), suppose $A_1$ is a sub-\loc algebra and $A_2$ is a \loc ideal. We prove by induction on $n\geq 0$ that
\begin{equation}
	\phi_1^{\star-1}({c})  = (\pi_1  \phi)({ c})\quad \tforall c\in  C_n.
\mlabel{eq:BHF3}
	\end{equation}
	
Notice  that $\phi(J)=\phi_1(J)=1_A$ implies $(\phi_1\,\star\, ( \pi_1  \phi))(J)=1_A$, so Eq.~\eqref{eq:BHF3} holds for $n=0$
since ${C}_0=  K   J$. Assuming  that Eq.~\eqref{eq:BHF3} holds for any $c$ in $ {C }$ of degree $\le n$, we prove that Eq.~\eqref{eq:BHF3} holds for any element $c\in  C_{n+1}$.

We write
\begin{equation*}
 \Delta (c) ={c}\ot J+ J\ot {c}+\sum_{({c})}{c}'\ot {c}''
\end{equation*}
with ${c}'\mtop_C {c}''$ of degree $\le n$.
 By the definition of $\phi_1 $,
  \begin{equation*} \phi_1({c})  =  -\pi_1(\phi(c)+  \sum_{({c})}\phi_1({c}^\prime)\, \phi({c}^{\prime\prime}))\end{equation*}
and
$\phi_1({c}^\prime)\mtop_A \phi({c}^{\prime\prime}).$
By the assumption on $\pi _i$, we have
$\phi_1({c}^\prime)\mtop_A \{\pi _1\phi({c}^{\prime\prime}), (\pi _2\phi)({c}^{\prime\prime})\}.$
Then
\begin{eqnarray*}
 \phi_1({c}) &=&-\pi_1\Big(\phi(c)- \sum_{({c})} \phi_1({c}^\prime) \big((\pi _1\phi)({c}^{\prime\prime})+(\pi _2 \phi)({c}^{\prime\prime})\big)\Big)\\
 &=&-(\pi_1\phi)(c)- \sum_{({c})} \phi_1({c}^\prime) (\pi _1\phi)({c}^{\prime\prime}).
\end{eqnarray*}
 Consequently,
\begin{equation*}
 \big(\phi_1 \, \star\, (\pi_1  \phi)\big)({c})=\phi_1({c})+\pi_1\phi(c)+ \sum_{({c})} \phi_1({c}^\prime) \,\pi_1  \phi({c}^{\prime\prime})=0.
\end{equation*}
 We conclude that $\phi_1\,\star\,\left(\pi_1 \phi\right)=e$, leading to $ \phi_1^{\star(-1)} = \pi_1  \phi$ since $\phi_1^{\star(-1)}$ has been shown to exist. {The {\loc algebraic Birkhoff factorization}  $\phi =\phi_1^{\star -1} \star \phi_2$  then yields  for $c{\in \ker}(\varepsilon)$:
 \begin{equation*}
  \pi_1(\phi(c)) + \phi_2(c) + \sum_{(c)}\pi_1(\phi(c'))\phi_2(c'') = \phi(c)
 \end{equation*}
  Using $\pi_2={\Id}-\pi_1$ gives   the recursive expression {\eqref{eq:phi2}} for $\phi_2$ in terms of $\pi_1$ and $\pi_2$.}
\end{proof}

\begin{rk}
	That the locality of $\phi$ implies that of $\phi_2$ and $\phi_1$ can be summarised under the motto {``}renormalisation preserves locality" in analogy with quantum field theory.
\mlabel{rk:motto}
\end{rk}

We illustrate this motto by applying Theorem~\mref{thm:abflcoalg} to the \loc linear map $S:\Q \bfc \to \calm _\Q$ that we have been taken as the main example throughout this paper. The map is denoted by $S^o$ in~\mcite{GPZ2} where it is shown, by the algebraic Birkhoff factorisation for connected coalgebras~\cite[Theorem~4.4]{GPZ2}, that both
$$S_1^{\star (-1)}: \Q \bfc \to \calm _{\Q,+} \ \text{and }\ S_2:\Q \bfc \to \calm _{\Q,-}$$
are linear maps sending the trivial lattice cone $(\{0\}, \{0\})$ to $1$ in $\calm_\Q$. Applying Theorem~\mref{thm:abflcoalg}, we further learn that both $S_1^{\star (-1)}$ and $S_2$ are \loc linear maps. This means that for cones $(C, \Lambda _C)$ and $(D, \Lambda _D)$ with $(C, \Lambda _C)\perp^Q (D, \Lambda _D)$, we also have
\begin{equation}
 S_1^{\star (-1)}(C, \Lambda _C)\perp ^Q  S_1^{\star (-1)}(D, \Lambda _D), \quad
S_2(C, \Lambda _C) \perp^Q  S_2(D, \Lambda _D).
\mlabel{eq:ls12}
\end{equation}
\section{Locality for Hopf algebras and the \Loc Product Conservation Principle}
\mlabel{sec:lhopf}

In this section, we combine a \loc algebra and a \loc coalgebra to {give} a \loc bialgebra, and then a \loc Hopf algebra under an additional connectedness condition. Together with the \loc Rota-Baxter algebra structure, we  obtain the \loc version of the algebraic Birkhoff factorisation originally given by Connes and Kreimer. This provides an answer to the question addressed in the \Loc Product Conservation Principle in Problem~\mref{pr:lphopf}. An application is given to the renormalisation of the exponential generating function of lattice cones.

\subsection{\Loc bialgebras and \loc Hopf algebras}
\mlabel{ss:lbialg}

\begin{defn}
\begin{enumerate}
\item
An {\bf \loc bialgebra} is a sextuple $( B , \top, m, u, \Delta, \vep)$ consisting of a \loc  algebra $( B , m, u, \top)$ and a \loc coalgebra $\left( B , \Delta, \top, \vep\right)$
that are \loc compatible in the sense that $\Delta$ and $\vep$ are \loc algebra homomorphisms.
\item
{A} \loc bialgebra $B$ is called {\bf connected} if there is a $\Z_{\geq 0}$-grading $B=\oplus_{n\geq 0} B_n$ with respect to which $B$ is both a \loc graded algebra in the sense of Definition~\mref{defn:localisedalgebra} and a connected \loc coalgebra in the sense of Definition~\mref{defn:colocalcoproduct}. {Then $J=1_B$.}
\end{enumerate}
\mlabel{defn:locbialgebra}
\end{defn}

Let us go back to the space $\Q \bfc$  of lattice cones with the Minkowski product and the coproduct $\Delta$ defined in Eq.~\eqref{eq:conecoprod}. We observe that the idempotence
$(C, \Lambda _C)\cdot (C, \Lambda _C)=(C, \Lambda _C)$
hinders the compatibility between the product and the coproduct. For example, taking $(C, \Lambda _C)=(\langle e_1\rangle, \Z e_1)$, then $\Delta(C\cdot C)=\Delta(C)\cdot \Delta(C)$ does not hold.
However, the following result shows that this compatibility can be recovered in the context of \loc bialgebras.

\begin {prop}
$(\Q \bfc,\perp^Q,\cdot,u,\Delta,\varepsilon) $ is a connected \loc bialgebra.
\mlabel{prop:conebialg}
\end{prop}

\begin {proof}
We first verify the compatibility of the \loc coalgebra structure and the \loc algebra structure.
To show that $\Delta$ is a \loc algebra homomorphism, we note that for $(C,\Lambda _C)\perp^Q (D, \Lambda _D)$, the faces of $(C,\Lambda _C)\cdot (D, \Lambda _D)$ are of the form
$(F_1, \Lambda _{F_1})\cdot (F_2, \Lambda _{F_2})$ with $F_1$ a face of $C$ and $F_2$ a face of $D$, and
$$t((C,\Lambda _C)\cdot (D, \Lambda _D), (F_1, \Lambda _{F_1})\cdot (F_2, \Lambda _{F_2}))=t((C,\Lambda _C), (F_1, \Lambda _{F_1}))\cdot t((D,\Lambda _D), (F_2, \Lambda _{F_2})).$$
So by definition, we have the desired equation:
\begin{equation}
\Delta ((C,\Lambda _C)\cdot (D, \Lambda _D))=\Delta (C,\Lambda _C)\cdot \Delta (D,\Lambda _D).
\mlabel{eq:coprodprod}
\end{equation}

The counit $\vep: \Q \bfc \to \Q$ is evidently an algebra homomorphism and hence a \loc algebra homomorphism.

Finally, we have checked that $\Q \bfc$ is both a graded \loc algebra (Lemma~\mref{lem:conegradalg}) and a connected \loc coalgebra (Lemma~\mref{lem:conecograd}). This completes the proof.
\end{proof}

\begin{defn}
A {\bf \loc Hopf algebra} is a \loc   bialgebra $\left( B ,\top, m, \Delta,u, \vep\right)$  with an antipode, defined to be a linear map $S:  B \to  B $
 such that $S$ and $\Id_ B  $ are mutually independent (in the sense of Definition \mref{defn:locallmap}) and
 \[S\star \Id=\Id\star S = u  \vep.\]
\mlabel {defn:LHopf}
\end{defn}

The usual proof (see e.g.\mcite{G2,M}) for the existence of the antipode on connected bialgebras extends to \loc bialgebras as follows. For $k\geq 1$, denote $m_1=m$ and $m_k=m(\id_B\ot m_{k-1})$.
\begin{lem}
 Let $\left(   B ,\top, m, u,\Delta,\vep\right)$ be  a connected \loc bialgebra, $\tilde\Delta^{\ot k}$ as in Eq.~(\mref{eq:redcoprod}) and $\alpha:   B  \to   B $ a \loc linear map with $\alpha (1_B)=0$. Then
\begin{enumerate}
\item
$\alpha^{\star k} = m_{k-1} \alpha^{\ot k} \tilde\Delta^{(k-1)}$ for all $k\geq 2$;
\mlabel{it:apower1}
\item
$\alpha^{\star k}(   B  _n)=\{0\}$ for all $k\geq n+1.$
\mlabel{it:apower2}
\end{enumerate}
\mlabel{lem:apower}
\end{lem}

\begin{proof}
(\mref{it:apower1}) We proceed by induction on $k\geq 2$ and first observe that $\alpha^{\star k}(1_B)=0$ for every $k\geq1$ as can easily be shown by induction. The result holds for $k=2$   since $\alpha(1_B)=0$. Assuming it holds at degree $k$ we write
  \begin{align*}
   \alpha^{\star k+1}(x) & = m(\alpha\ot\alpha^{\star k})\Delta(x) \\
			 & = m(\alpha\ot\alpha^{\star k})\tilde\Delta(x) \quad\text{(since }\alpha^{\star k}(1_B)=\alpha(1_B)=0)\\
			 & = m\big(\alpha\ot (m_{k-1}\alpha^{\ot k} \tilde\Delta^{(k-1)})\big)\tilde\Delta(x) \\
& = m(\id_B \ot m_{k-1})(\alpha\ot \alpha^{\ot k})(\id_B\ot \tilde{\Delta}^{(k-1)})\tilde{\Delta}(x) \\
			 & = m_k\alpha^{\ot (k+1)}\tilde\Delta^{(k)}(x)
  \end{align*}
 where we have used the \loc property of $\alpha$ and the fact that $m$ is associative.
\smallskip

\noindent
(\mref{it:apower2}) is a direct consequence of (\mref{it:apower1}) and Lemma~\mref{lem:conil}.(\mref{it:conil2}).
\end{proof}

\begin{prop}
		 Let $\left(   B  ,\top , m, u,\Delta,\vep\right)$ be  a connected \loc bialgebra.
There is  a linear map $S: {   B  }\to {   B }$  with the properties of the antipode stated above. It is given by
		  \begin{equation*}
		   S=\sum_{k=0}^\infty (u  \vep -\Id)^{\star k}.
		  \end{equation*}
\mlabel{prop:localisedantipode}
\end{prop}

\begin{proof}
The map $\alpha:   B  \to   B $  defined by $\alpha=\Id-u  \vep$  is \loc linear, and $\alpha(1_B)=0$. The von Neumann series   $S=\sum_{k=0}^{\infty}(-1)^k\alpha^{\star k}$  which is locally
  finite by Lemma~\mref{lem:apower}.(\mref{it:apower2}) and hence well-defined, gives  the inverse of the identity for the convolution product.
\end{proof}

As an immediate consequence of Propositions~\mref{prop:conebialg} and ~\mref{prop:localisedantipode}, we have
\begin{coro}
The \loc bialgebra $(\Q \bfc,\perp^Q,\cdot,u,\Delta,\varepsilon)$ is a \loc Hopf algebra.
\mlabel{co:conehopf}
\end{coro}

\subsection{Locality of the convolution of \loc algebra homomorphisms}
 		\begin{prop}Let $\left( B ,\mtop_ B , m, u,  \Delta, \vep \right)$   be a \loc bialgebra. Let $\left(A,\mtop_A,\cdot\right) $ be a \loc commutative  algebra. Let
 			\begin{equation*}
 			\phi,\psi: \left( B ,\mtop_ B  \right)\longrightarrow   \left(A ,\mtop_A\right)
 			\end{equation*}
 			be independent \loc linear maps.
 		  \begin{enumerate}
 			\item If $\phi$ and $  \psi$ are \loc multiplicative then so is their convolution product $\phi\star \psi$.
 			\item Assume further that $ B $ is connected. If $\phi$ is a homomorphism of \loc algebras, then  so is its convolution inverse $ \phi^{\star (-1)}$. So the set $\calg$ of homomorphisms of \loc algebras from $(B,\mtop_B)$ to $(A,\mtop_A)$ is a \loc group with respect to the independent relation of \loc linear maps.
 	\end{enumerate}
\mlabel{prop:localityconv}
 	\end{prop}

 \begin{proof}
(i) For $( c,  d)\in \mtop_ B $, by the proof of Lemma~\mref{lem:coalgmap}.(\mref{it:coalgmap3}), we can write
$$\Delta (c)=\sum_i c_{i1}\ot c_{i2}, \quad \Delta (d)=\sum_j d_{j1}\ot d_{j2}$$
with
$(c_{i1}, c_{i2}, d_{j1}, d_{j2})\in B^{_\top 4}.$
Then
$$\Delta (cd)=(m\ot m)\tau _{23}(\Delta \ot \Delta )(c\ot d)=\sum_{i,j} c_{i1}d_{j1}\ot c_{i2}d_{j2}.
$$
So
\begin{eqnarray*}
 			(\phi\star \psi) ( c\,  d)&=&
 			\sum_{i,j} \phi\left(c_{i1}d_{j1}\right)\,\psi\left( c_{i2}d_{j2}\right) \\
 			&=&\sum_{i,j} \phi\left(  c_{i1} \right)\, \phi\left( d_{j1}\right)\,\psi\left(  c_{i2}\right)\,\psi\left(  d_{j2}\right) \\
 			&=&\sum_{i} \phi\left(  c_{i1}\right)\,\psi\left(  c_{i2}\right)\, \sum_{j} \phi\left( d_{j1}\right)\,\psi\left(  d_{j2}\right) \\
 			&=& \left(\phi\star \psi( c)\right)\,\left(\phi\star \psi( d)\right).
 			\end{eqnarray*}

\noindent
(ii) Now we use induction on the sum of degrees of $c$ and $d$, $c\top d$ to prove
$$\phi^{\star(-1)}(c)\phi^{\star(-1)}(d)=\phi^{\star(-1)}(cd),
$$
which is true if the sum of degrees is 0.

In general, by Lemma~\mref{lem:conil}.(\mref{it:conil0}), we write
$$\Delta (c)=c\ot J +J\ot c+\sum _{(c)}c'\ot c^{\prime \prime}
, \quad
\Delta (d)=d\ot J +J\ot d+\sum _{(d)}d'\ot d^{\prime \prime}
$$
with
$(c', c'', d', d'')\in B^{_\top 4}.$
So by $\Delta(cd)=\Delta(c)\Delta(d)$, we obtain
\begin {eqnarray*}
\Delta (cd)&=&cd\ot J+J\ot cd +c\ot d+d\ot c+\sum _{(d)}cd'\ot d^{\prime \prime}+\sum _{(d)}d'\ot cd^{\prime \prime}\\
&&+\sum _{(c)}c'd\ot c^{\prime \prime}+\sum _{(c)}c'\ot c^{\prime \prime}d+
\sum _{(c)(d)}c'd'\ot c^{\prime \prime}d''.
\end{eqnarray*}
By Eq.~(\mref {eq:phirec}) we obtain
\begin{eqnarray*}\phi ^{\star(-1)}( cd  )&=&-\phi (cd)-\phi (c)\phi ^{\star(-1)}(d)-\phi (d)\phi ^{\star(-1)}(c)\\
&&-\sum _{(d)}\phi (cd')\phi ^{\star (-1)}( d^{\prime \prime})-\sum _{(d)}\phi (d')\phi ^{\star (-1)}( cd^{\prime \prime})\\
&&-\sum _{(c)}\phi (c'd)\phi ^{\star (-1)}( c^{\prime \prime})-\sum _{(c)}\phi (c')\phi ^{\star (-1)}( c^{\prime \prime}d)\\
&&-\sum _{(c)(d)}\phi (c'd')\phi ^{\star (-1)}( c^{\prime \prime}d'').
\end{eqnarray*}

By Eq.~(\mref {eq:phirec}) applied to $c$ and $d$, the \loc multiplicativity of $\phi $, the commutativity of $A$ and induction hypothesis, we have
\begin{eqnarray*}\phi ^{\star(-1)}( cd  )&=&\phi(c)\phi (d)+\sum _{( d )}\phi (c)\phi ( d  ^\prime) \phi ^{\star(-1)}(  d  ^{\prime \prime})+\sum _{( c )}\phi ( c  ^\prime)\phi (d) \phi ^{\star(-1)}(  c  ^{\prime \prime})\\
&&+\sum _{( c )(d)} \phi ( d  ^\prime) \phi ^{\star(-1)}(  d  ^{\prime \prime})\phi ( c  ^\prime) \phi ^{\star(-1)}(  c  ^{\prime \prime})
+\sum _{( c )(d)}\phi ( c  ^\prime) \phi ^{\star(-1)}(  c  ^{\prime \prime})\phi ( d  ^\prime) \phi ^{\star(-1)}(  d  ^{\prime \prime})\\
&&-\sum _{(c)(d)}\phi (c'd')\phi ^{\star (-1)}( c^{\prime \prime}d'')\\
&=&\phi(c)\phi (d)+\sum _{( d )}\phi (c)\phi ( d  ^\prime) \phi ^{\star(-1)}(  d  ^{\prime \prime})+\sum _{( c )}\phi ( c  ^\prime)\phi (d) \phi ^{\star(-1)}(  c  ^{\prime \prime})\\
&&+\sum _{( c )(d)}\phi ( c  ^\prime) \phi ^{\star(-1)}(  c  ^{\prime \prime})\phi ( d  ^\prime) \phi ^{\star(-1)}(  d  ^{\prime \prime})\\
&=& \big(\phi ( c )+\sum _{( c )}\phi ( c  ^\prime) \phi ^{\star(-1)}(  c  ^{\prime \prime})\big)\big(\phi ( d )+\sum _{( d )}\phi ( d  ^\prime) \phi ^{\star(-1)}(  d  ^{\prime \prime})\big)\\
&=&\phi ^{\star(-1)}( c  )\phi ^{\star(-1)}( d).
\end{eqnarray*}
This completes the induction.
\end{proof}

\subsection{The \Loc Product Conservation Principle}
\mlabel{ss:lphopf}
Now we give the locality of the algebraic Birkhoff factorisation in the Hopf algebra context~\mcite{CK}.
\begin{theorem}
{\bf (Algebraic Birkhoff  factorisation, \loc Hopf algebra version)}
Let $\left( H ,\mtop_H \right)$   be a \loc connected Hopf algebra, $H=\oplus_{n\geq 0} H _n$, $H _0= K  e$.
Let $\left(A,\mtop_A,\cdot, P \right) $ be a commutative \loc Rota-Baxter algebra of weight -1 with $P$ idempotent.
Denote $A_1=P(A)$ and $A_2=(\Id - P)(A)$.
Let
\begin{equation*}
 \phi: \left(H , \mtop_H \right)\longrightarrow   \left(A,\mtop_A\right)
\end{equation*}
be a \loc algebra homomorphism.
Then there are unique independent \loc algebra homomorphisms $\phi_i: H\to K+A_i$ with $\phi_i (\ker \vep)\subseteq A_i, i=1,2, $ such that

\begin{equation}
\phi= \phi_1^{\star (-1)} \star \phi_2.
\mlabel{eq:lhabf}
\end{equation}
The map $  \phi_1^{\star(-1)}$ is also a \loc algebra homomorphism and $\phi_1\mtop_\call \phi$, $\phi_1\mtop_\call \phi_2$.

If in addition $A_2$ is a \loc ideal of $A$, then $\phi_1^{\star(-1)}=\pi_1 \phi$
and $\phi_2$ is recursively given by
$$ \phi_2(1_H)=1_A, \quad
\phi_2(c)=(\pi_2  \phi)(c)-\sum_{(c)}(\pi_1  \phi)(c')\phi_2(c'') \tforall c\in \ker \vep.
$$
with $\tilde\Delta(c)=\sum_{(c)}c'\ot c''$ the reduced coproduct defined in Lemma \ref{lem:reduced_coproduct}.
\mlabel{thm:abflhopf}
\end{theorem}
\begin{rk}
Theorem~\mref{thm:abflhopf} provides an answer to Problem~\mref{pr:lphopf}: the \loc multiplicativity of the renormalised map follows from that of the initial map under the assumption of the theorem.
\end{rk}

\begin {proof}
All the statements follows from Theorem~\mref{thm:abflcoalg} except the claims that $\phi_i, i=1,2$ and $\phi_1^{\star(-1)}$ are \loc algebra homomorphisms.

For $c\mtop_H d$, by Lemma~\mref{lem:conil}.(\mref{it:conil0}), we can write
$$\Delta (c)=c\ot {1_H} +{1_H}\ot c+\sum _{(c)}c'\ot c^{\prime \prime}, \quad
\Delta (d)=d\ot {1_H} +{1_H}\ot d+\sum _{(d)}d'\ot d^{\prime \prime}
$$
with
$(c', c'', d', d'')\in H^{_\top 4}.
$
Using the \loc Rota-Baxter property of $P$, by a similar argument as in the non-\loc case~\cite[Theorem~2.4.3]{G2}, we can prove  that $\phi _1$ and $\phi _2$ are homomorphisms of \loc algebras. By Proposition~\mref {prop:localityconv}, $ \phi_1^{\star(-1)}$ is
a homomorphism of \loc algebras.
\end{proof}

Applying Theorem \mref{thm:abflhopf} to $(A,\mtop_A, \cdot) =\left(\calm _\Q, \perp^Q,\cdot\right)$  yields the following result.

\begin{coro} Let $(H,\mtop_H)$   be a connected \loc Hopf algebra.
	Let
	\begin{equation*}
	\phi: \left( H  , \mtop_ H  \right)\longrightarrow   \left(\calm _\Q, \perp^Q\right)
	\end{equation*}
	be a \loc  linear map  such that $\phi(1_H)=1_{\mathcal M _\Q}$.
	Let $\phi=(\phi_1^Q)^{\star (-1)}\star \phi_2^Q$ be the algebraic Birkhoff  factorisation in Eq.~\eqref{eq:lhabf} with $\phi_1^Q(1_H)=\phi_2^Q(1_H)=1_{\mathcal M_\Q}$. Then
	\begin{enumerate}
\item
$\pi_1^Q \phi$ is a \loc linear map;
\item $(\phi_1^Q)^{\star (-1)}=\pi_1^Q  \, \phi$  so that
\begin{equation} \phi=(\pi_1^Q\,  \phi)\star \phi_2^Q;
\mlabel{eq:BHlocalpi}
\end{equation}
\item  the maps $\phi_1^Q$, $\phi_2^Q$ are \loc linear maps  and $\phi_1^Q\mtop_\call \phi$,  $\pi_1^Q\, \phi\mtop_\call \phi_2^Q$;
\item assuming in addition that  $\phi$ is a \loc algebra homomorphism, then the maps
    $\pi_1^Q\phi, \phi_1^Q$ and $\phi_2^Q$ are \loc algebra homomorphisms.
		\end{enumerate}
\mlabel{co:abflhopf}
\end{coro}
\begin{proof}  The proof is straightforward; let us nevertheless mention that $\pi_1^Q\, \phi\mtop_\call \phi_2^Q$ follows from $\phi_1\mtop_\call \phi_2$ in Eq.~(\mref{eq:lhabf}) combined with the fact that the convolution inverse preserves locality.
\end{proof}

We end the paper by applying Corollary~\mref{co:abflhopf} to the \loc algebra homomorphism $S:\Q \bfc \to \calm_\Q$. The classical algebraic Birkhoff factorisation does not apply while Theorem~\mref{thm:abflcoalg} does, telling us that
$$S_1^{\star (-1)}: \Q \bfc \to \calm _{\Q,+}\ \text{and }\ S_2:\Q \bfc \to \calm _{\Q,-}$$
are \loc linear maps. From Corollary~\mref{co:abflhopf}, we conclude that the two maps are also \loc algebra homomorphisms and thus are multiplicative for orthogonal pairs of lattice cones. Noting further~\mcite{GPZ2} that $S_2=I$, the exponential integral and $S_1^{\star (-1)}=\mu$, the interpolation factor in the Euler-Maclaurin formula
$ S=\mu\star I,$
we obtain the following consequence of Theorem~\mref{thm:abflhopf}, Proposition~\mref{pp:merodecomp} and Corollary~\mref{co:abflhopf}.

\begin{coro}
For any orthogonal pair of lattice cones $(C,\Lambda _C)$ and $(D,\Lambda _D)$, we have
$$ \mu((C,\Lambda _C)\cdot (D,\Lambda _D))=\mu(C,\Lambda _C)\mu (D,\Lambda _D), \quad I((C,\Lambda _C)\cdot (D,\Lambda _D))=I(C,\Lambda _C)I(D,\Lambda _D).$$
Further, $\mu=\pi_+^Q S$.
\mlabel{co:muiloc}
\end{coro}

\noindent
{\bf Acknowledgements.} The authors acknowledge supports from the Natural Science Foundation of China (Bin Zhang: Grants No. 11071176, 11221101 and 11521061, Li Guo: Grant No. 11371178) and  the German Research Foundation (Sylvie Paycha and Pierre Clavier: DFG grant PA 1686/6-1). They are grateful to the hospitalities of Sichuan University and University of Potsdam where parts of the work were completed.

\end{document}